\documentclass[11pt]{article}

\usepackage{etex}
\usepackage[thinlines]{easytable}
\usepackage{algorithmic, algorithm}
\usepackage{amsmath,amsfonts,amssymb,tikz}
\usepackage[margin=1in]{geometry}
\usepackage{amsthm} 
\usepackage{csquotes}
\usepackage[all]{xy}
\usepackage[bottom]{footmisc}
\usepackage{enumitem}
\setlist{parsep = -0em, itemsep = 0.25em}
\usepackage{prettyref} 
\usepackage{hyperref}
\usepackage{bm, bbm}
\usepackage{romannum}

\newtheorem{theorem}{Theorem}[section]

\newtheorem{conj}[theorem]{Conjecture}
\newtheorem{lemma}[theorem]{Lemma}
\newtheorem{corollary}[theorem]{Corollary}

\newtheorem{claim}[theorem]{Claim}

\newtheorem{definition}[theorem]{Definition}

\newtheorem{fact}[theorem]{Fact}

\makeatletter
\newtheorem*{rep@theorem}{\rep@title}
\newcommand{\newreptheorem}[2]{%
\newenvironment{rep#1}[1]{%
\def\rep@title{#2 \ref{##1}}%
\begin{rep@theorem}}%
{\end{rep@theorem}}}
\makeatother

\newreptheorem{theorem}{Theorem}
\newreptheorem{lemma}{Lemma}
\newreptheorem{definition}{Definition}

\newenvironment{proofsketch}{\noindent{\bf Proof Sketch.}}%
{\hspace*{\fill}$\Box$\par}
\newenvironment{proofof}[1]{\smallskip\noindent{\bf Proof of #1.}}%
{\hspace*{\fill}$\Box$\par}

\newcommand{\pref}{\prettyref}
\newrefformat{lem}{Lemma \ref{#1}}
\newrefformat{cl}{Claim \ref{#1}}
\newrefformat{prop}{Proposition \ref{#1}}
\newrefformat{thm}{Theorem \ref{#1}}
\newrefformat{cor}{Corollary \ref{#1}}
\newrefformat{cha}{Chapter \ref{#1}}
\newrefformat{sec}{Section \ref{#1}}
\newrefformat{tab}{Table \ref{#1}}
\newrefformat{fig}{Figure \ref{#1}}
\newrefformat{conj}{Conjecture \ref{#1}}
\newrefformat{prot}{Protocol \ref{#1}}
\newrefformat{fact}{Fact \ref{#1}}
\newrefformat{def}{Definition \ref{#1}}

\newcommand{\E}{{\mathbb{E}}}
\newcommand{\eps}{\varepsilon}

\newcommand{\cB}{\mathcal{B}}

\newcommand{\cP}{\mathcal{P}}

\newcommand{\cM}{\mathcal{M}}

\newcommand{\bl}{\bm{\ell}}
\newcommand{\bj}{\bm{j}}

\newcommand{\Patrascu}{P\u{a}tra\c{s}cu}

\usepackage{xspace}

\newcommand{\poly}{{\operatorname{poly}\xspace}}

\newcommand{\SEL}[2]{\mathsf{SEL}^{#2}_{#1}\xspace}

\newcommand{\AND}{\mathsf{AND}}
\newcommand{\DISJ}{\mathsf{DISJ}}
\newcommand{\KL}[2]{\mathsf{D}_{KL} ( {#1} || {#2} )}

\newcommand{\EC}[1]{}

\newcommand{\Pat}{P\v{a}tra\c{s}cu}

\floatstyle{boxed}
\newfloat{Protocol}{H}{pro}

\title{An Adaptive Step Toward the Multiphase Conjecture}

\author{Young Kun Ko\thanks{NYU Courant, {\tt ykk254@cs.nyu.edu}} 
\and Omri Weinstein\thanks{Columbia University. Email: \texttt{omri@cs.columbia.edu}. 
Research supported by NSF CAREER award CCF-1844887.}}
\date{\today}

\begin{document}

\maketitle
\thispagestyle{empty}
\begin{abstract}

In 2010, \Pat~proposed the following three-phase dynamic problem, 
as a candidate for proving \emph{polynomial} lower bounds on the operational time of dynamic data structures:  

\begin{itemize}
\item  \bf \Romannum{1}\rm: Preprocess a collection of sets $\vec{S} = S_1, \ldots , S_k \subseteq [n]$, where $k=\poly(n)$.  
\item  \bf \Romannum{2}\rm: A set $T\subseteq [n]$ is revealed, and the data structure updates its memory. 
\item  \bf \Romannum{3}\rm: An index $i \in [k]$ is revealed, and the data structure must determine if $S_i\cap T=^? \emptyset$.  
\end{itemize}

\Pat~conjectured that any data structure for the Multiphase problem must make $n^\eps$ cell-probes 
in either Phase \Romannum{2} or \Romannum{3}, and showed that this would imply similar \emph{unconditional} lower bounds on 
many important dynamic data structure problems. Alas, there has been almost no progress on this conjecture 
in the past decade since its introduction. 
We show an $\tilde{\Omega}(\sqrt{n})$ cell-probe lower bound on the Multiphase problem 
for data structures with general (adaptive) updates, and queries with 
unbounded but ``layered" adaptivity. 
This result captures all known set-intersection data structures  and significantly strengthens previous 
Multiphase lower bounds, which only captured non-adaptive data structures.

Our main technical result is a communication lower bound on a 4-party variant of P\v{a}tra\c{s}cu's \emph{Number-On-Forehead}  Multiphase game, 
using information complexity techniques.  We also show that a lower bound on P\v{a}tra\c{s}cu's original NOF game would 
imply a polynomial ($n^{1+\eps}$) lower bound on the number of wires of any 
constant-depth circuit with \emph{arbitrary} gates computing a random $\tilde{O}(n)\times n$ \emph{linear} 
operator $x \mapsto Ax$, a long-standing open problem in circuit complexity. 
This suggests that the NOF conjecture is much stronger than its data structure counterpart. 

\end{abstract}

\newpage

\pagenumbering{arabic}
\section{Introduction}

Proving unconditional lower bounds on the operational time of dynamic data structures has 
been a challenge since the introduction of the \emph{cell-probe} model \cite{Yao79}. 
In this model, the data structure needs to support a sequence of $n$ online updates and queries, 
where the operational cost is measured only by the number of \emph{memory accesses} (``probes") 
the data structure makes to its memory, whereas all computations on probed cells are completely 
free of charge. A natural question to study is the tradeoff between the \emph{update} time $t_u$ 
and \emph{query} time $t_q$ of the data structure for supporting the underlying dynamic problem.  
Cell-probe lower bounds provide a compelling answer to this question, as they are 
purely information-theoretic and independent of implementation constraints, hence apply to 
any imaginable data structure. Unfortunately, the abstraction of the cell-probe model also comes 
at a price, and  the highest explicit lower bound known to date,  on \emph{any} dynamic 
problem, is merely polylogarithmic ($\max\{t_u, t_q\} \geq \tilde{\Omega}(\lg^2 n)$, see e.g., \cite{Lar12, LWY17} 
and references therein). In 2010, \Pat~\cite{patrascu:multiphase} proposed the following dynamic set-disjointness problem, 
known as the \emph{Multiphase problem}, 
as a candidate for proving \emph{polynomial} lower bounds on the operational time of dynamic 
data structures. 
The problem proceeds in 3 ``phases":

\begin{itemize}
\item \bf P\Romannum{1}\rm: Preprocess a collection of $k=\poly(n)$ sets $\vec{S} = S_1, \ldots , S_k \subseteq [n]$. 
\item \bf P\Romannum{2}\rm: A set $T\subseteq [n]$ is revealed, and the data structure can update its memory in $O( n t_u )$ time. 
\item \bf P\Romannum{3}\rm: An index $i \in [k]$ is revealed, and the data structure must determine if $S_i \cap T =^? \emptyset$ 
in $O(t_q)$-time.
\end{itemize}

\Pat~conjectured that any data structure solving  the Multiphase problem must make $\max\{t_u, t_q\}\geq n^{\eps}$ cell-probes, 
and showed that such a polynomial lower bound would imply similar polynomial lower bounds on many important dynamic 
data structure problems, including dynamic reachability in directed graphs and online matrix multiplication (for the broad implications 
and further context of the Multiphase conjecture within fine-grained complexity, see \cite{patrascu:multiphase, Henzinger2015}). 
In the same paper, \Pat~\cite{patrascu:multiphase} 
proposed an approach to prove an unconditional cell-probe lower bound on the Multiphase problem, 
by reduction to the following 3-party \emph{Number-On-Forehead} (NOF) communication game $\SEL{\DISJ_n}{k}$, henceforth 
called the \emph{Multiphase Game} : 
\begin{itemize}
\item Alice receives a collection of sets $\vec{S} = S_1, \ldots , S_k \subseteq [n]$ and a random index $i\in_R [k]$. 
\item Bob receives  a set $T\subseteq [n]$ and the index $i$.  
\item Charlie receives $\vec{S}$ and $T$ (but not $i$). 
\end{itemize}
Thus, one can think of $i$ as being on Charlie's forehead, $T$ being on Alice's forehead, and $\vec{S}$ being 
on Bob's forehead.  The goal of the players is to determine if $S_i\cap T =^? \emptyset$, where communication proceeds in the 
following way: First, Charlie sends a message (``advice") $U=U(\vec{S}, T)$ \emph{privately to Bob}. Thereafter, Alice 
and Bob continue to compute $\DISJ_n (S_i,T)$ in the standard 2-party model. Denoting by $\Pi$ the 
second stage protocol,  \Pat~made the following conjecture: 

\begin{conj}[Multiphase Game Conjecture, Conjecture 9 of \cite{patrascu:multiphase}] \label{MPH_NOF_conj}
Any 3-party NOF protocol for the Multiphase game with $|U| =o(k)$ bits of advice must have $|\Pi| > n^\eps$ communication. 
\end{conj}
  
The (na\"{i}ve) intuition for this conjecture 
is that, since Charlie's advice is independent of $i$, it can only provide 
very little useful information about the interesting subproblem $\DISJ_n (S_i,T)$ (assuming $|U|=o(k)$), 
and hence Alice and Bob might as well solve the problem in the standard 2-party model. 
This intuition turns out to be misleading, and in fact when $S_i$'s and $T$ are \emph{correlated},  
it is simply false -- Chattopadhyay, Edmonds, Ellen and Pitassi~\cite{littleadvice} showed a deterministic (2-round) 
NOF protocol for the Multiphase game with a total of $O(\sqrt{n}\log k) = \tilde{O}(\sqrt{n})$ 
communication, 
whereas the 2-party communication complexity of set-disjointness is $\Omega(n)$, even 	
randomized. Surprisingly, they also show that Conjecture~\ref{MPH_NOF_conj} is equivalent, 
up to $O(\log k)$ communication factor, for deterministic and randomized protocols. 
Nevertheless, Conjecture~\ref{MPH_NOF_conj} still stands  
for \emph{product distributions} \cite{littleadvice} (incidentally,  the communication complexity 
of set-disjointness under product distributions is $\tilde{\Theta}(\sqrt{n})$ \cite{BFS86, Hastad2007}). 

The technical centerpiece of this paper is an $\Omega(\sqrt{n})$ lower bound on the NOF Multiphase game, 
for (unbounded-round)  protocols in which \emph{only the first} message of Alice in $\Pi$  
depends on her entire input $\vec{S} = S_1, \ldots , S_k$ (and $i$), while in subsequent rounds $j>1$, Alice's messages 
can depend only \emph{on $S_i$, $i$ and the transcript $\Pi^{<j}$} (No restriction is placed on Bob and Charlie). 
Note that Alice's messages in subsequent rounds 
still heavily depend on \emph{all sets} $\vec{S}$, but only through the 
\emph{transcript} of $\Pi=\Pi(\vec{S})$ so far (this feature better captures 
data structures, which can only adapt based on cells probed so far).  
There is a natural way to view such restricted 3-party NOF protocol in terms 
of an additional player (Megan, holding $\vec{S},i$),   
who, in addition to Charlie's private advice to Bob, can broadcast a single message to \emph{both} Alice and Bob 
(holding $S_i, T$ respectively), 
who then continue to communicate in the standard 2-party model (see \pref{fig:4party}). 
We define this \emph{4-party NOF} model formally  
in \pref{sec:nof}. Our main technical result is the following lower bound on such 
NOF protocols:

\begin{theorem}[Informal]  
\label{thm:multiphase_informal}
Any Restricted NOF protocol $\Gamma=(U,\Pi)$ for the Multiphase game with $|U| =o(k)$ 
bits of advice must have $|\Pi| > \Omega(\sqrt{n})$ communication. 
\end{theorem} 

This lower bound is tight up to logarithmic factors, as the model generalizes the upper bound of 
\cite{littleadvice} (See \pref{sec:appendixLB}).   
This suggests that the NOF model we study is both subtle and powerful. Indeed, while the aforementioned 
restriction may seem somewhat technical, we show that removing it by allowing as little as \emph{two rounds} 
of Alice's messages to depend on her entire input $\vec{S}$, would lead to a major breakthrough 
in circuit lower bounds -- see \pref{thm:ciruitwire_informal} below. 
Interestingly, the Multiphase conjecture itself does not have this implication, 
since dynamic data structures only have limited and local access to $\vec{S}$, through the probes (``transcript") 
of the query algorithm, and hence induce weaker NOF protocols.

\paragraph{Implications to dynamic data structure lower bounds}

In contrast to the \emph{static} cell-probe model, \emph{adaptivity} plays a dramatic role when it comes to dynamic data structures. 
In \cite{BL15}, Brody and Larsen consider a variation of the Multiphase problem with $(\lg n)$-bit updates (i.e., the 2nd phase set is of 
cardinality $|T|=1$), and show that any dynamic data structure whose query algorithm is \emph{non-adaptive}\footnote{An algorithm 
is non-adaptive, if the addresses of probed memory cells are predetermined by the query itself, and do not depend on content 
of cells probed along the way.} must make $\max\{t_u,t_q\} \geq \Omega(n/w)$ 
cell-probes when the word-size is $w$ bits. Nevertheless, such small-update problems have a trivial ($t_q=O(1)$) 
adaptive upper bound
and therefore 
are less compelling from the prospect of making progress on 
general lower bounds. (By contrast, proving polynomial
cell-probe lower bounds for dynamic problems with \emph{large}  $\poly(n)$-bit updates, like Multiphase, 
even against non-adaptive query algorithms,
already seems beyond the reach of current techniques
\footnote{While intuitively larger updates $|T|=\poly(n)$ only make the problem harder and should therefore 
only be easier to prove lower bounds against, the \emph{total} update time of the 
data structure in Phase \Romannum{2} is also proportional to $|T|$ and hence the data structure has 
potentially much more power as it can ``amortize" its operations. This is why encoding-style 
arguments fail for large updates (enumerating all ${n \choose |T|}=\exp(n)$ possible updates is prohibitive).}).

We prove a polynomial lower bound  on the Multiphase problem, against a much stronger class 
of data structures, 
which we call \emph{semi-adaptive}, defined as follows:  

\begin{definition}[Semi-Adaptive Data Structures] \label{def:semiadaptive}
Let $D$ be a dynamic data structure for the Multiphase problem with general (adaptive) updates. 
Let $\cM(\vec{S})$ denote the memory state of $D$ after the preprocessing Phase \Romannum{1}, 
and let $\Delta(\cM, T)$ denote the set of ($ \leq |T| \cdot t_u$) cells updated in Phase \Romannum{2}. 
$D$ is \emph{semi-adaptive} if its query algorithm in Phase \Romannum{3} operates 
in the following stages (``layers"): 
\begin{itemize}
\item Given the query $i\in [k]$, $D$ may first read $S_i$ \emph{free of charge}.  
\item $D$ (adaptively) reads at most $t_1$ cells from $\cM$. 
\item $D$ (adaptively) reads at most $t_2$ cells from $\Delta(\cM, T)$, and returns the answer $S_i\cap T=^? \emptyset$.  
\end{itemize} 
The update time of $D$ is $t_u$, and the query time is $t_q := t_1+t_2$. 
\end{definition}

Thus, the query algorithm has unbounded but ``layered" adaptivity in Phase \Romannum{3},  
as the model allows only a single alternation between the two layers of memory cells, $\cM$ and $\Delta(\cM, T)$.  
While this restriction may seem somewhat technical, all known set-intersection data structures are special 
cases of the semi-adaptive model (see \cite{Demaine00adaptiveset, BK02, BY04, 
BPP07, CP10, Kopelowitz2015} and references therein). In fact, this model 
allows to solve the Multiphase problem in $t_u=t_q = \tilde{O}(\sqrt{n})$ time (see \pref{sec:appendixUB}), 
indicating that the model is very powerful. 
We remark that even though the set of modified cells $\Delta(\cM, T)$ may be \emph{unknown} to $D$ 
at query time, it is easy to implement a semi-adaptive data structure by maintaining $\Delta(\cM, T)$ in 
a dynamic dictionary \cite{WPP05} that checks membership of cells in $\Delta$, and returns $\perp$ if 
the cell is $\notin \Delta(\cM, T)$. 

Our main result is an essentially tight 
$\tilde{\Omega}(\sqrt{n})$ cell-probe lower bound on the Multiphase problem against  
semi-adaptive data structures. 
This follows from \pref{thm:multiphase_informal} by a simple variation of the reduction 
in~\cite{patrascu:multiphase} : 

\begin{theorem}[Multiphase Lower Bound for Semi-Adaptive Data Structures] \label{thm_dynamic_multiphase_informal}
Let $k > \omega( n )$. Any semi-adaptive data structure that solves the Multiphase problem, 
must have either $n \cdot t_u \geq \Omega(k/w)$ or $t_q \geq \Omega(\sqrt{n}/w)$, in the dynamic cell-probe model with word size $w$.
\end{theorem}


\paragraph{Implications of the NOF game to Circuit Lower Bounds}

A long-standing open problem in circuit complexity is whether \emph{non-linear} gates can 
help computing linear operators (\cite{lupanov56,JS10,Drucker12}). 
Specifically, the challenge is to prove a polynomial ($n^{1+\eps}$) lower bound on the number 
of \emph{wires} of constant-depth circuits with \emph{arbitrary gates} for computing \emph{any}
linear operator $x\mapsto Ax$ \cite{Jukna_2012}. 
A random matrix $A$ easily gives a polynomial ($n^2/\lg n$) lower bound against 
\emph{linear circuits} \cite{lupanov56, JS10} (which restricted the interest to finding \emph{explicit} 
hard matrices $A$, see \cite{Val77}). 
In contrast, for general circuits, the highest lower bound on the number of wires,  
even for computing a \emph{random} matrix $A$, is near-linear \cite{Drucker12, GHKPV13}. 

Building on a recent reduction of Viola \cite{V18}, we show that \Pat's NOF Conjecture~\ref{MPH_NOF_conj}, 
even for \emph{3-round protocols}, would already resolve this question. 
This indicates that Conjecture~\ref{MPH_NOF_conj} is in fact much stronger than the Multiphase conjecture itself.

\begin{theorem}[NOF Game Implies Circuit Lower Bounds]   \label{thm:ciruitwire_informal}
Suppose Conjecture~\ref{MPH_NOF_conj} holds, even for 3-round protocols. Then for $m = \omega( n )$,
there exists a linear \footnote{Over the boolean Semiring, i.e., where addition are replaced with 
$\vee$ and multiplication are replaced with $\wedge$. We note that there is evidence that 
computing $Ax$ over the boolean Semiring is easier 
than over $\mathbb{F}_2$ \cite{Gronlund2015}, hence in that sense our lower bound is stronger than over finite fields.}
operator $A \in \{ 0 , 1 \}^{m \times n }$ such that any depth-$d$ circuit 
computing $x \mapsto A x$  (with \emph{arbitrary gates and unbounded fan-in}) requires $n^{1 + \Omega \left( \eps /d \right)}$ wires. 
In particular, if $d = 2$, the conjecture implies that computing $A x$ for some $A$ requires $n^{1 + \frac{\eps}{2} - o(1)}$ wires.
\end{theorem}


\paragraph{Comparison to previous work.}

The aforementioned work of Brody and Larsen \cite{BL15} proves essentially optimal ($\Omega(n/w)$) dynamic lower bounds 
on variations of the Multiphase problem, when either the update or query algorithms are \emph{nonadaptive} (or 
in fact ``memoryless" in the former, which is an even stronger restriction). Proving lower bounds in the semi-adaptive model is a different 
ballgame, as the $\tilde{O}(\sqrt{n}/w)$ upper bound suggests (\pref{sec:appendixUB}). 
We also remark that \cite{BL15} were the first to observe a (similar but different) connection between nonadaptive data structures and 
depth-2 circuit lower bounds.  

A more recent result of Clifford et. al \cite{Gronlund2015} proves a ``threshold" cell-probe lower 
bound on \emph{general} dynamic data structures solving the Multiphase problem, asserting that fast queries $t_q = o(\log k / \log n)$ require 
very high $t_u > k^{1-o(1)}$ update time. Unfortunately, this result does not rule out data structures with 
$t_u=t_q = \poly\log(n)$ or even $O(1)$ time 
for the Multiphase problem (for general data structures, neither does ours).

As far as the Multiphase NOF Game, the aforementioned work of Chattopadhyay et. al~\cite{littleadvice} proves a tight $\tilde{\Theta}(\sqrt{n})$ 
communication lower bound against so-called \emph{``1.5-round"} protocols, 
in which Bob's message to Alice is \emph{independent of the index $i$}, hence he is essentially ``forwarding" a small ($o(n)$) portion of 
Charlie's message to Alice (this effectively eliminates Bob from communicating, making it similar to a 2-party problem). 
While our restricted NOF model is formally incomparable to ~\cite{littleadvice} (as in our model, Alice is the first speaker),  
\pref{thm:multiphase_informal} in fact subsumes it by a simple modification (see Appendix \ref{sec:appendixLB}).  
The model we study seems fundamentally stronger than 1.5-protocols, as it inherently 
involves multiparty NOF communication.

To best of our knowledge, all previous lower bounds ultimately reduce the Multiphase problem to a \emph{2-party} communication 
game, which makes the problem more amenable to \bf compression-based \rm arguments.  This is the main departure point 
of our work. 
We remark that most of our information-theoretic tools in fact apply to general dynamic data structures.  
We discuss this further in \pref{sec:openproblem} at the end of this paper.


\section{Technical Overview} 
Here we provide a streamlined overview of our main technical result, \pref{thm:multiphase_informal}. 
As discussed earlier in the introduction, a na\"{i}ve approach to the Multiphase Game $\SEL{\DISJ_n}{k}$ is 
a ``round elimination" approach:  
Since Charlie's advice consists of only $|U| = o( k )$ bits and he has no knowledge of the index of the interesting subproblem 
$i\in_R [ k ]$, his advice $U$ to Bob should convey $o(1)$ bits of information about the interesting set $S_i$ and hence 
Alice and Bob might hope to simply ``ignore" his advice $U$ and use such efficient NOF protocol $\Gamma$ 
to generate a too-good-to-be-true  \emph{2-party} protocol for set disjointness (by somehow ``guessing" Charlies's 
message which appears useless, and absorbing the error). The fundamental flaw with this intuition is that Charlie's advice is a function 
of \emph{both} players' inputs, hence conditioning on $U(\vec{S},T)$ \emph{correlates} the 
inputs in an arbitrary way, extinguishing the standard ``rectangular" (Markov) property of 2-party protocols in the second phase interaction 
$\Pi^{A\leftrightarrow B}$ between Alice and Bob (This is the notorious feature preventing ``direct sum" arguments in 
NOF communication models). 
In particular, Chattopadhyay et. al~\cite{littleadvice} show that a small advice ($|U| = \tilde{O}(\sqrt{n})$) can already 
decrease the communication complexity of the multiphase problem to at most the  2-party complexity of set-disjointness 
under \emph{product} distributions, yielding a surprising 2-round $\tilde{O}(\sqrt{n})$ upper bound on the Multiphase game. 
(This justifies why our hard distribution for $\SEL{\DISJ_n}{k}$ will be a product distribution to begin with, i.e., $\vec{S} \perp T$). 
Alas, even if the inputs are originally independent ($I(\vec{S};T)=0$), they may 
\emph{not} remain so throughout $\Pi$, and it is generally possible that $I(S_i;T | \Pi) \gg 0$. 
This means that, unlike 2-party protocols,  
$\Pi = \Pi(U,\vec{S},T)$ \emph{introduces correlation} between the inputs, and as such, is not amenable to the 
standard analysis of 2-party communication techniques.  

Nevertheless, one might still hope that if the advice $U$ is small enough, then this correlation will be small 
for an average index $S_i$ when the inputs are independently chosen. 
At a high level, our proof indeed shows that if only the \emph{first} message of Alice can (directly) depend on 
her entire input $\vec{S} = S_1,\ldots, S_k$ (whereas her subsequent messages $\Pi^\tau$ are only a function of 
$S_i$, $i$ and the transcript history $[\Pi(\vec{S},T,i)]^{<\tau}$), then  
it is possible to \emph{simultaneously} control the information cost and correlation of $\Pi$, so long as the advice 
$U$ is small enough ($o(k)$). This in turn facilitates a ``robust" direct-sum style argument for \emph{approximate protocols}. 
More formally, our proof consists of the following two main steps: \\ 

\noindent $\bullet$ \bf A low correlation and information random process for computing $\AND$. \rm  
The first part of the proof shows that an efficient Restricted NOF protocol $\Gamma$ for the Multiphase game $\SEL{\DISJ_n}{k}$ 
(under the natural \emph{product} distribution on $\vec{S},T$) can be used to design a certain random process $Z(X,Y)$ 
computing the 2-bit AND function (on 2 independent bits $\sim \cB_{\Theta(1/\sqrt{n})}$), which \emph{simultaneously} 
has \emph{low information cost} w.r.t $X,Y$ and \emph{small correlation}, meaning that the input bits remain 
roughly independent at any point during the process, i.e,  $I(X;Y|Z) = o(1/n)$. Crucially, $Z$ is \emph{not} a valid 
2-party protocol (Markov chain) -- if this were the case, then we would have $I(X;Y|Z) = 0$ since a deterministic 2-party 
protocol cannot introduce any additional correlation  between the original inputs (this is also the essence of the the celebrated 
\emph{``Cut-and-Paste" Lemma} \cite{CutPaste}). Nevertheless, 
we show that for restricted NOF protocols $\Gamma$ (equivalently, unrestricted protocols in our 4-party model, cf. Figure \ref{fig:4party}),
it is possible to design such random variable $Z(X,Y)$ from $\Gamma$, which is \emph{close enough} to a 
Markov chain. The design of $Z$ requires a careful choice conditioning variables as well as a ``coordinate sampling" 
step for reducing entropy, though the analysis of this 
part uses standard tools (the chain rule and subadditivity of mutual information).We first design a $Z'$ 
with similar guarantees for single-copy \emph{disjointness}, and then use (a variation of) the  standard direct-sum 
information cost argument to ``scale down" the information and correlation of $Z$ so as to extract the desired random 
process for 2-bit AND. An important observation in this last step is that the direct sum property of information cost 
holds not just for communication protocols, but in fact for general random variables. \\ 

\noindent $\bullet$ \bf A ``robust" Cut-and-Paste Lemma. \rm 
The second part is proving that such random variable $Z(X,Y)$ cannot exist, i.e., ruling out a random process $Z(X,Y)$ 
computing $\AND$ (with 1-sided error under $X,Y \sim^{iid} \cB_{\Theta(1/\sqrt{n})}$) which simultaneously has 
\emph{low information cost} and \emph{small correlation} ($o(1/n)$). 
The high-level intuition is that, if $Z(X,Y)$ introduces little correlation, then the distribution over $X$ and $Y$ conditioned 
on $Z(X,Y)$ should remain \emph{approximately} a product distribution, i.e., close to a rectangle. By the correctness guarantee of $Z$, 
the distribution on $\{0,1\}^2$ conditioned on $Z(X,Y) = \AND(X,Y) = 0$ must have 0 mass on the $(1,1)$ entry. 
But if this conditional distribution does not contain $(1,1)$ in its support and \emph{close to a rectangle}, 
a KL-divergence calculation shows that $Z(X,Y)$ must reveal a lot of information about 
either $X$ or $Y$ (this calculation crucially exploits the fact that Pinsker's inequality is  
loose ``near the ends", i.e.,  $\KL{p}{q} \approx \|p-q\|_1$ for $p,q=o(1)$, and there is no quadratic loss). 
Our argument can be viewed as a generalization of the Cut-and-Paste Lemma to more robust settings of 
random variables (``approximate protocols"). We remark that while the proof of the original Cut-and-Paste 
Lemma \cite{CutPaste} heavily relies on properties of the Hellinger distance, this technique does not seem to easily extend to 
small-correlation random variables. This forces us to find a more direct argument, 
which may be of independent interest.

\section{Preliminaries}
\subsection{Information Theory}

In this section, we provide necessary backgrounds on information theory and information complexity that are used 
in this paper. For further reference, we refer the reader to \cite{CoverThomas}.

First we define entropy of a random variable, which intuitively quantifies how ``random" a given random variable is.

\begin{definition}[Entropy]
	The entropy of a random variable $X$ is defined as 
	\begin{equation*}
	H(X) := \sum_{x} \Pr[X=x] \log \frac{1}{\Pr[X=x]}.
	\end{equation*}
	Similarly, the conditional entropy is defined as
	\begin{equation*}
	H(X|Y) := \E_{Y} \left[ \sum_{x} \Pr[X = x | Y = y] \log \frac{1}{\Pr[ X = x | Y = y]} \right].
	\end{equation*}
\end{definition}

\begin{fact}[Conditioning Decreases Entropy] \label{fact:conditioningentropy}
For any random variable $X$ and $Y$
\begin{equation*}
    H(X) \geq H(X|Y)
\end{equation*}
\end{fact}

\noindent With entropy defined, we can also quantify correlation between two random variables, or how much information one random variable conveys about the other.

\begin{definition}[Mutual Information]
	Mutual information between $X$ and $Y$ (conditioned on $Z$) is defined as 
	\begin{equation*}
	I( X; Y | Z) := H(X|Z) - H(X|YZ).
	\end{equation*}
\end{definition}

\noindent Similarly, we can also define how much one distribution conveys information about the other distribution.

\begin{definition}[KL-Divergence]
	KL-Divergence between two distributions $\mu$ and $\nu$ is defined as
	\begin{equation*}
	\KL{\mu}{\nu} := \sum_{x} \mu(x) \log \frac{\mu(x)}{\nu(x)}.
	\end{equation*}
\end{definition}

\noindent In order to bound mutual information, it suffices to bound KL-divergence, due to following fact.

\begin{fact}[KL-Divergence and Mutual Information] \label{fact:KL_Mutual}
	The following equality between mutual information and KL-Divergence holds 
	\begin{equation*}
	I(A;B|C) = \E_{B,C} \left[ \KL{ A|_{B=b, C=c} }{ A|_{C=c} } \right].
	\end{equation*}
\end{fact}



\noindent We can expect the following inequality near 0 distributions. 

\begin{fact}[Divergence Bound] \label{fact:divergence}
If $\KL{ B_q }{ B_p } < o(p)$ with $p = o(1)$, then $q \in [ 0.99 p , 1.01 p ]$
\end{fact}
\begin{proof}
Since $B_q$ is decreasing as $q$ goes to $p$, we show that if $q = 0.99p$ or $q = 1.01p$, $D(B_q || B_p) \geq \Omega(p)$. First if $q = 1.01 p$, then the term is
\begin{align*}
    &D(B_q || B_p ) = 1.01 p \log 1.01 + \left( 1 - 1.01p \right) \log \frac{1 - 1.01 p}{1 - p} \\
    & \geq 0.01449 p + \left( 1 - 1.01 p \right) \log \left( 1 - \frac{0.99p}{1-p} \right) \\
    & \geq 0.01449 p + \left( 1 - 1.01 p \right) ( - 0.01443 p ) \geq \Omega(p)
\end{align*}
Similarly, if $q = 0.99 p$, we have
\begin{align*}
    &D(B_q || B_p ) = 0.99 p \log 0.99 + \left( 1 - 0.99p \right) \log \frac{1 - 0.99 p}{1 - p} \\
    & \geq -0.01436 p + \left( 1 - 0.99 p \right) \log \left( 1 - \frac{0.99p}{1-p} \right) \\
    & \geq -0.01436 p + \left( 1 - 0.99 p \right) ( 0.01442 p ) \geq \Omega(p).
\end{align*}

\end{proof}

\noindent We also make use of the following facts on Mutual Information throughout the paper.

\begin{fact}[Chain Rule] \label{fact:chainrule}
For any random variable $A,B,C$ and $D$  
\begin{equation*}
    I(AD;B|C) = I(D;B|C) + I(A;B|CD).
\end{equation*}
\end{fact}

\begin{fact} \label{fact:chainrule1}
For any random variable $A,B,C$ and $D$, if $I(B;D|C) = 0$
\begin{equation*}
    I(A;B|C) \leq I(A;B|CD).
\end{equation*}
\end{fact}
\begin{proof} By the chain rule and non-negativity of mutual information, 
\begin{align*}
I(A;B|C) \leq I(AD;B|C) = I(B;D|C) + I(A;B|CD) = I(A;B|CD).
\end{align*}
\end{proof}

\begin{fact}\label{fact:chainrule2}
For any random variable $A,B,C$ and $D$, if $I(B;D|AC) = 0$
\begin{equation*}
    I(A;B|C) \geq I(A;B|CD).
\end{equation*}
\end{fact}
\begin{proof} By the chain rule and non-negativity of mutual information,  
\begin{align*}
I(A;B|CD) \leq I(AD;B|C) = I(A;B|C) + I(B;D|AC) = I(A;B|C).
\end{align*}
\end{proof}

\subsection{NOF Communication Models} \label{sec:nof}

\begin{figure}[h!]
    \centering
    \includegraphics[scale=0.44]{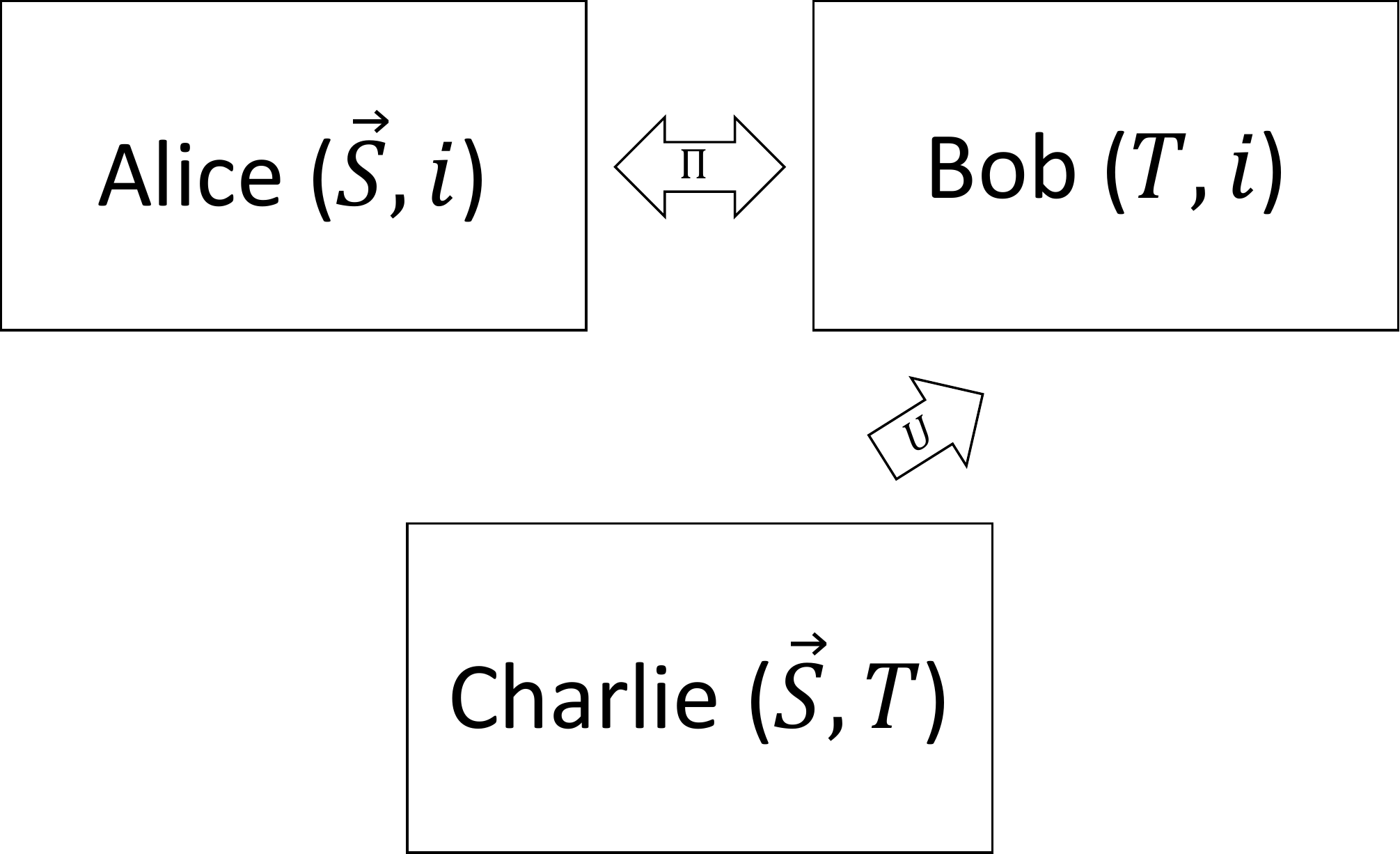}
    \caption{3-party NOF Communication}
    \label{fig:model} 
\end{figure}

In \Pat's NOF Multiphase Game $\SEL{f}{k}$, there are 3 players with the following 
information on their foreheads: 
\noindent \bf Charlie: \rm an index $i\in [k]$ ; 
\bf Bob: \rm a collection of sets $\vec{S} := S_1,\ldots, S_k \subseteq [n]$ ; 
\bf Alice: \rm  a set $T\subseteq [n]$. 
I.e., Charlie has access to both $\vec{S}$ and $T$, but not to $i$. Alice has access to $\vec{S}$ and $i$, and Bob has access to 
$T$ and $i$. The goal is to compute $$\SEL{f}{k} := f(S_i,T).$$ 
The communication proceeds as follows: 
In the first stage of the game, Charlie sends a message (``\emph{advice}") $U= U(\vec{S},T)$ \emph{privately to Bob}.  
In the second stage, Alice and Bob continue to communicate in the standard 2-party settting to compute $f (S_i, T)$ 
(see Figure \ref{fig:model}). 
We denote such protocol by $\Gamma := (U, \Pi_i)$ where $\Pi_i$
is the second stage transcript, assuming the index of the interesting subproblem is $i$.

Unfortunately, lower bounds for general protocols in \Pat's 3-party NOF model seem  beyond the reach of current techniques, 
as we show in \pref{sec:consequences} that Conjecture \ref{MPH_NOF_conj}, even for \emph{3-round} protocols,  
would resolve a major open problem in circuit complexity. Fortunately, for dynamic data structure applications, 
weaker versions of the NOF model suffice (this is indeed one of the enduring messages of this paper). 

We consider the following restricted class of protocols. We say that $\Gamma = (U,\Pi)$ is a {\bf restricted} NOF 
protocol if  Alice is the first speaker in $\Pi$ (in the second stage of the game) and only her \emph{first} message $\Pi_i^{1}$ to Bob 
depends on her entire input $\vec{S}$ and $i$,   whereas in subsequent rounds, Alice's messages $\Pi^\tau$ may depend only on 
$S_{i}, i$ and the history of the transcript 
$\Pi^{<\tau}$ with Bob. Note that the latter means that Alice and Bob's subsequent messages can still heavily depend on $S_{-i}$, but only through 
the transcript (this feature better captures data structures, since the query algorithm can only adapt based on the information 
in cells it already probed, and not the entire memory). 

\begin{figure}[h!]
    \centering
    \includegraphics[scale=0.5]{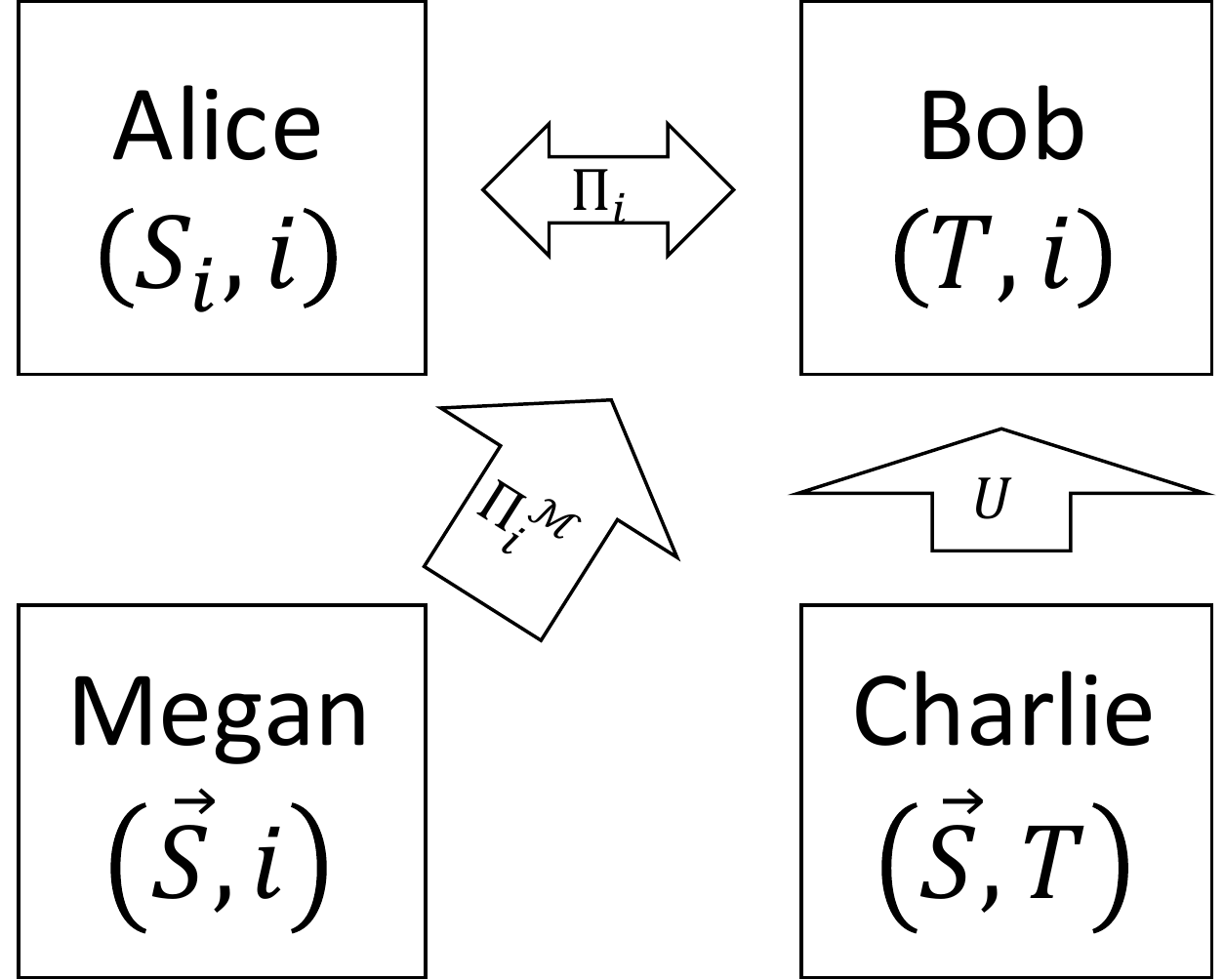}
    \caption{4-party NOF Communication}
    \label{fig:4party} 
\end{figure}

\paragraph{An equivalent 4-party NOF model.} Restricted 3-party NOF protocols are more naturally described by the following 4-party NOF model. 
{\bf Alice} has access only to $S_i$ and $i$, {\bf Bob} has access to $T$ and $i$. {\bf Charlie} has access to $\vec{S}$ and $T$, but no access to $i$. {\bf Megan} has access to $\vec{S}$ and $i$. 
In the first stage of the protocol $\Gamma$, in addition to Charlie's \emph{private} message to Bob $U(\vec{S},T)$, 
Megan can \emph{broadcast} a message $\Pi^\cM = \Pi^\cM(\vec{S},i)$ to both Alice and Bob. Thereafter, Alice and Bob proceed 
to communicate  in the 2-party model as before, denoted $\Pi^{A \leftrightarrow B}$. See Figure \ref{fig:4party}. 
We denote 4-party protocols by $\Gamma = (U, \Pi)$ where $\Pi := ( \Pi^\cM,  \Pi^{A \leftrightarrow B})$.  
We write $\Pi_i := ( \Pi_i^\cM,  \Pi_i^{A \leftrightarrow B})$ to denote the transcript of $\Pi$ when 
the index of the interesting subproblem is $i\in [k]$.  

It is straightforward to see that Restricted 3-party NOF protocols for the Multiphase Game are equivalent to (unrestricted) 4-party protocols 
(by setting Alice's first message as Megan's message $\Pi_i^\cM$, and Charlie remains unchanged). 
As such, our main technical theorem (\pref{thm:multiphase_informal}) can be restated as follows. 
\begin{theorem}[4-party NOF Lower Bound] \label{thm:4party} 
Let $k > \omega( n )$. For any 4-party NOF protocol $\Gamma = (U,\Pi)$ that solves $\SEL{\DISJ_n}{k}$  
with $|U| < o ( k )$, there exists $i \in [k]$ such that $|\Pi_i| > \Omega( \sqrt{n})$.	
\end{theorem}





\section{Lower Bound for 4-party NOF Protocols}

\paragraph{Notations.}
We denote by $S_i^j$ the $j$-th entry of the set $S_i$ and analogously for $T$. We write 
$S_i^{<j} :=  S_i^1, \ldots , S_i^{j-1}$, similarly  $S_i^{-j} := S_i^1, \ldots , S_i^{j-1}, S_i^{j+1} , \ldots, S_i^n$, 
and analogously for $T$. 

For technical reasons, we shall need to carry out the proof on a restricted subset $\cP$ of the original $[k]$ coordinates, where 
$\cP = (i_1, \ldots , i_p ) \in [k]^p$. We write $i \in \cP$ if there exists some index $\ell \in [p]$ such that $i_\ell = i$. We write $(i_1, \ldots, i_{\ell - 1})$ and $(i_1, \ldots, i_{\ell })$ in short hand as $i_{< \ell}$ and $i_{\leq \ell}$ respectively. $S_{i_{< \ell}}$ refers to $S_{i_1} , \ldots , S_{i_{\ell - 1}}$ and $\Pi_{i_{< \ell}}^{\cM}$ refers to $\Pi^{\cM}_{i_1}, \ldots , \Pi^{\cM}_{i_{\ell - 1}}$. Furthermore, $S_{\cP}$ and $\Pi^{\cM}_{\cP}$ refers to $S_{i_1} , \ldots , S_{ i_{p} }$ and $\Pi^{\cM}_{i_1} , \ldots , \Pi^{\cM}_{ i_{p} }$ respectively. Also $\Pi_{i}^{ans}$ denotes the output of $\Gamma = (U, \Pi)$ when the index of interest is $i \in [k]$.

Let $C := \max_{i\in [k]} |\Pi_i|$ be the maximal number of 
bits exchanged between Megan, Alice and Bob over all $i \in [k]$. Then in particular, for every $i\in [k]$, 
\begin{equation} \label{eq:smallq}
|\Pi^{\cM}_i | \leq |\Pi_i| \leq C.
\end{equation}
Observe that since Megan does not have access to $T$, for any subset $\cP$ of coordinates it holds that 
\begin{equation} \label{eq:non-adaptive}
    I(T; S_{\cP} \Pi^{\cM}_\cP ) = 0
\end{equation}
assuming $\vec{S} \perp T$, since Megan's message only depends on $i$ and $\vec{S}$. Indeed, 
$I(T; S_{\cP} \Pi^{\cM}_\cP ) \leq I(T ; \vec{S}, \Pi^{\cM}_{\cP} ) = I(T ; \vec{S} ) + I(T; \Pi^{\cM}_\cP ~ | \vec{S}) = 0$. 
It is noteworthy that, by contrast, $I(T; \Pi^{\cM}_\cP ~ | \vec{S}, \bf U\rm) \neq 0$, since conditioning on $U$ correlates Megan's message 
with Bob's input. Indeed, dealing with this subtle feature will be the heart of this section and will later explain the choice of $Z^\DISJ$.  


\paragraph{Hard Distribution.} We consider the natural hard product distribution for set-disjointness, extended to  
$\SEL{\DISJ_n}{k}$:  
For all $i \in [k]$ and $j \in [n]$, $S_i^j$ and $T^j$ are i.i.d. Bernoulli $\cB_\gamma$ for $\gamma = \frac{1}{1000 \sqrt{n}}$. 

\subsection{A Low Correlation Random Process for \texorpdfstring{$\DISJ_n$}{DISJ}}

The goal of this section is to show that an efficient 4-party NOF protocol $\Gamma$ for $\SEL{\DISJ_n}{k}$ 
implies a \emph{low-correlation, low-information} random process for computing a single copy of set-disjointness (under the 
hard product distribution). 
For technical reasons, we restrict the proof to a random subset $\cP = (I_1, \ldots, I_p) \in_R [k]^p$ of $p$ coordinates, 
with the constraint that for any $k_1, k_2 \in [p]$, if $k_1 \neq k_2$, then $I_{k_1} \neq I_{k_2}$ where $p$ is a parameter 
that will be chosen as $o(k)$.\footnote{This is equivalent to picking a random ordering over a random $p$-sized subset in $[k]$} Let $\bl \in_R [p]$ be a uniformly random index. 
We shall prove the following Lemma: 

\begin{lemma} \label{lem:goodq}
Let $\Gamma = (U, \Pi)$ be a 4-party NOF protocol for $\SEL{\DISJ_n}{k}$ with $|\Pi_i| < C$ for all $i \in [k]$. Then 
for $p = o( k )$, there exists a random variable $Z^{\DISJ}$ containing $\cP$ and $\bl$ such that
\begin{itemize}
\item If $\DISJ_n(S_{I_{\bl}},T) = 0$, then $Z^{\DISJ}_{ans} = 0$.
\item If $\DISJ_n(S_{I_{\bl}},T) = 1$, then $Z^{\DISJ}_{ans} = 1$. 
\item Satisfies the following information cost bound
\begin{align}
    & I( Z^{\DISJ} T ; S_{I_{\bl}} ) \leq C  + o( C )  \label{eq:smallS} \\
    & I( Z^{\DISJ} ; T )  \leq C \label{eq:smallT} 
\end{align}
\item Satisfies the following correlation bound	
\begin{align}
     I( S_{i_{\ell}} ; T | Z^{\DISJ} ) \leq O \left( \frac{|UT|}{p} \right). \label{eq:smallCorrelation}
\end{align} 
\end{itemize}
\end{lemma}


Intuitively, \pref{lem:goodq} states that an efficient 4-party NOF protocol for $\SEL{\DISJ_n}{k}$ can be used to design 
a random process $Z(S_{I_{\bl}}, T)$ which for a random $\cP$ and $\bl$ computes $\DISJ_n$ on inputs $S_{I_{\bl}}$ and $T$, 
in a way that simultaneously: (i) $Z$ reveals small information on average about both $S_{I_{\bl}}$ and $T$, and (ii) creates 
small correlation between $S_{I_{\bl}}$ and $T$ assuming $|UT| = o(p) \leq o( k )$ (i.e., it is in some sense ``close" to a 2-party 
communication protocol). The choice of $Z$ is set to 
\begin{align*}
& Z^{\DISJ} := \Pi_{I_{\bl}} S_{I_{< \bl}} \Pi_{I_{< \bl}}^{\cM} ,\cP, \bl = \Pi_{I_{\bl}} \Pi_{I_{\bl}}^{\cM} S_{I_{< \bl}} \Pi_{I_{< \bl}}^{\cM} ,\cP, \bl \\
& Z^{\DISJ}_{ans} := \Pi_{I_{\bl}}^{ans}.
\end{align*}
where the equality holds for $Z^{\DISJ}$ since $\Pi_{I_{\bl}}^{\cM}$ is included in $\Pi_{I_{\bl}}$. Since we do not bound the number of rounds, we can without loss of generality assume that $\Pi_{I_{\bl}}^{ans}$ is included in $\Pi_{I_{\bl}}$.
Note that $\Pi_{I_{\bl}}, \Pi_{I_{< \bl}}^{\cM}$ are random variables that depend on Charlie's advice 
$U$, but importantly $U$ is not included explicitly in $Z^{\DISJ}$. 
We begin with the following claim, which morally states that $Z^{\DISJ}\setminus \Pi_{I_{\bl}}$ 
reveals little information on an average set $S_i$: 

\begin{claim} \label{cl:pathsampling}
\begin{equation} \label{eq:smallinfoonS}
    I( S_{I_{\bl}} ; S_{I_{< \bl}} \Pi_{I_{< \bl}}^{\cM} T , \cP , \bl ) \leq \frac{p \cdot C}{k - p}. 
\end{equation}
\end{claim}
\begin{proof}
First, note that we have
\begin{equation*} 
	I( S_{I_{\bl}} ; S_{I_{< \bl}} \Pi_{I_{< \bl}}^{\cM} , \cP, \bl ) = I( S_{i_{\ell}}  ; S_{i_{< \ell}} \Pi_{i_{< \ell}}^{\cM} | \cP \bl  ) = \E_{ \cP, \bl } \left[ I( S_{i_\ell} ; S_{i_{< \ell}} \Pi_{i_{< \ell}}^{\cM} | \cP = i_{\leq p}, \bl = \ell) \right]
\end{equation*}
where the first equality holds since $I(S_{I_{\bl}}; \cP, \bl ) = 0$ since $\cP, \ell$ are independent of $\vec{S}$ and $T$ along with \pref{fact:chainrule}. Furthermore, for any setting of $\cP$ and $\bl$, we have
\begin{equation*}
I( S_{i_\ell} ; S_{i_{< \ell}} \Pi_{i_{< \ell}}^{\cM} | \cP = i_{\leq p}, \bl = \ell) = I( S_{i_\ell} ; S_{i_{< \ell}} \Pi_{i_{< \ell}}^{\cM}).
\end{equation*}
since the choice of $\cP$ and $\bl$ are independent of entries in $\vec{S}$ and $T$.
Therefore, we have
\begin{equation} \label{eq:nopl}
	I( S_{I_{\bl} } ; S_{I_{< \bl}} \Pi_{I_{< \bl}}^{\cM} \cP \bl ) =I( S_{I_{\bl} } ; S_{I_{< \bl}} \Pi_{I_{< \bl}}^{\cM} | \cP \bl ) = \E_{\cP, \bl} \left[ I( S_{i_{\ell} } ; S_{i_{< \ell}} \Pi_{i_{< \ell}}^{\cM} ) \right]. 
\end{equation}
Now consider fixed $\bl = \ell$ and $I_{< \ell} = i_{< \ell}$. Then we prove the following inequality.
\begin{equation} \label{eq:p_average}
	\E_{\substack{ i_\ell, \\ \forall r < \ell, i_\ell \neq i_{r}}} \left[ I( S_{i_\ell} ; S_{i_{< \ell}} \Pi_{i_{< \ell}}^{\cM}, T ) \right] \leq \frac{p \cdot C}{k - p}.
\end{equation}
First observe that 
\begin{align*}
	& I( S_{i_\ell} ; S_{i_{< \ell}} \Pi_{i_{< \ell}}^{\cM}, T | I_{< \ell} = i_{<\ell} , I_{\geq \ell } ) = I( S_{i_\ell} ; S_{i_{< \ell}} \Pi_{i_{< \ell}}^{\cM}, T | I_{\ell}, I_{< \ell} = i_{<\ell}  ) \\
	& = \E_{i_\ell} \left[ I( S_{i_\ell} ; S_{i_{< \ell}} \Pi_{i_{< \ell}}^{\cM}, T | I_{\ell} = i_\ell, I_{< \ell} = i_{<\ell} ) \right] = \E_{\substack{ i_\ell, \\ \forall r < \ell, i_\ell \neq i_{r}}} \left[ I( S_{i_\ell} ; S_{i_{< \ell}} \Pi_{i_{< \ell}}^{\cM}, T ) \right]
\end{align*}
Now we have that for all $i \in [k]$ such that $i \notin (i_1, \ldots , i_{\ell -1})$, $S_i$'s are i.i.d. Therefore we get 
\begin{align*}
	& \E_{\substack{ i_\ell, \\ \forall r < \ell, i_\ell \neq i_{r}}} \left[ I( S_{i_\ell} ; S_{i_{< \ell}} \Pi_{i_{< \ell}}^{\cM}, T ) \right] = \frac{1}{k - (\ell - 1)} \sum_{i \notin i_{< \ell}} I(S_i ; S_{i_{< \ell}} \Pi_{i_{< \ell}}^{\cM}, T ) \\
	&  \leq \frac{I(S_{i \notin i_{< \ell}} ; S_{i_{< \ell}} \Pi_{i_{< \ell}}^{\cM}, T ) }{k - (\ell - 1)} = \frac{I(S_{i \notin i_{< \ell}} ; \Pi_{i_{< \ell}}^{\cM}| S_{i_{< \ell}}, T ) }{k - (\ell - 1)} \leq \frac{H(\Pi_{i_{< \ell}}^{\cM} | S_{i_{< \ell}}, T)}{k - (\ell - 1)} \\
	& \leq \frac{H( \Pi_{i_{< \ell}}^{\cM} )}{k - (\ell - 1)} \leq \frac{(\ell - 1) C }{k - (\ell - 1)} 
\end{align*}
where the last equality holds since $I(S_{i \notin i_{< \ell}} ; S_{i_{< \ell}}, T ) = 0$ from our assumption on the hard distribution.
Now since we have $\ell \leq p$, we get 
\begin{equation*}
    \E_{\substack{ i_\ell, \\ i_\ell \notin i_{< \ell}}} \left[ I( S_{i_\ell} ; S_{i_{< \ell}} \Pi_{i_{< \ell}}^{\cM}, T ) \right] \leq \frac{(\ell - 1) C }{k - (\ell - 1)}  \leq  \frac{p \cdot C}{k - p}.
\end{equation*}
Therefore, we have that \pref{eq:p_average} holds for any fixed $\bl = \ell$ and $I_{< \ell}= (i_1, \ldots, i_{\ell - 1})$. Taking expectation over $\cP$ and $\bl$, we get 
\begin{align*}
\E_{\cP, \bl} \left[ I( S_{i_{\ell} } ; S_{i_{< \ell}} \Pi_{i_{< \ell}}^{\cM} T ) \right]	= \E_{\bl} \left[  \E_{I_{< \ell}} \E_{I_{\ell}} \left[ I( S_{i_{\ell} } ; S_{i_{< \ell}} \Pi_{i_{< \ell}}^{\cM} T ) \right] \right] \leq \frac{p \cdot C}{k - p} .
\end{align*}
\end{proof}

The next claim, which is another direct application of the chain rule,  
asserts that for a random coordinate $i\in \cP$, Megan's messages in $\Pi_i$ ($\Pi_i^{\cM}$) do not heavily depend 
on $T,U$, conditioned on previous coordinate transcripts. 

\begin{claim} \label{cl:smallcorr}
For any fixed $\cP = i_{\leq p}$, if $\bl$ is uniformly distributed over $[p]$
\begin{equation} \label{eq:ut}
\E_{l} \left[ I( S_{i_\ell} \Pi_{i_{\ell}}^{\cM} ; U T | S_{i_{< \ell}} \Pi_{i_{< \ell}}^{\cM} , \bl = \ell, \cP = i_{\leq p}  \right] \leq O \left( \frac{|UT|}{p} \right)
\end{equation}
\end{claim}
\begin{proof}
Again, since $\bl$ is picked independently at random ($UT$ is independent of $\ell$), we get
\begin{equation*}
I( S_{i_\ell} \Pi_{i_{\ell}}^{\cM} ; U T | S_{i_{< \ell}} \Pi_{i_{< \ell}}^{\cM} , \bl = \ell, \cP = i_{\leq p} ) = 	I( S_{i_\ell} \Pi_{i_{\ell}}^{\cM} ; U T | S_{i_{< \ell}} \Pi_{i_{< \ell}}^{\cM} , \cP = i_{\leq p} )
\end{equation*}
Then taking expectation over $\bl$, we get
\begin{align*}
    & \E_{\bl} \left[ I( S_{i_\ell} \Pi_{i_{\ell}}^{\cM} ; UT | S_{i_{< \ell}} \Pi_{i_{< \ell}}^{\cM}, \cP = i_{\leq p} ) \right] \\
    & = \frac{1}{p} \sum_{\ell \in [p]} I( S_{i_\ell} \Pi_{i_{\ell}}^{\cM} ; UT | S_{i_{< \ell}}\Pi_{i_{< \ell}}^{\cM} , \cP = i_{\leq p} ) \\
    & = \frac{I( S_{i_{\leq p}} \Pi^{\cM}_{i_{\leq p}} ; UT | \cP = i_{\leq p} ) }{p} \leq  \frac{|UT|}{p}.
\end{align*}
	
\end{proof}

Recall that $\Pi_{i} := (\Pi_{i}^{\cM}, \Pi_{i}^{A \leftrightarrow B})$ is the transcript between Megan, Alice and Bob when the index of the interesting subproblem is $i$. We now turn to establish the fact that conditioning on 
$\Pi_{i}$ cannot introduce too much correlation between the (originally independent) $S_{i}$ and $T$. As discussed in the 
introduction, if $\Pi_i$ 
were a standard (deterministic) 2-party protocol, then this would have indeed been the case (as the \emph{rectangle} 
property of communication protocols ensures that independent inputs $S_i,T$ remain so throughout the 
protocol: $I(S_i;T|\Pi_i) = 0$). Alas, $\Pi_i$ no longer has the rectangle property anymore (as Charlie's message $U(\vec{S},T)$ correlates the 
inputs in an arbitrary way). Fortunately, we will be able to show that if Megan's messages only depend on $\vec{S}$ and $i$, and Alice's response only depend on $S_i, i$ and previous transcript then we can 
control the correlation introduced by $Z^\DISJ$ 
(by adding the aforementioned extra variables in the definition of $Z^\DISJ$) \emph{without increasing the information cost} 
of $Z^\DISJ$ with respect to $S_i$ and $T$. We begin with the following claim, 
which shows that the effect of conditioning on $Z^\DISJ$ 
can be upper bounded by the following term: 

\begin{claim} \label{cl:roundelimination}
For any fixed $\ell \in [p]$, and $\cP = i_{\leq p}$
\begin{align*}
    & I(S_{i_{\ell}} ; T | \Pi_{i_\ell} S_{i_{< \ell}} \Pi_{i_{< \ell}}^{\cM}, \cP = i_{\leq p}, \bl = \ell) = I(S_{i_{\ell}} ; T | \Pi_{i_\ell}^{\cM} \Pi_{i_\ell}^{A \leftrightarrow B}, S_{i_{< \ell}} \Pi_{i_{< \ell}}^{\cM}, \cP = i_{\leq p}, \bl = \ell) \\
    & \leq I(S_{i_\ell}\Pi_{i_\ell}^{\cM} ; UT | S_{i_{< \ell}} \Pi_{i_{< \ell}}^{\cM},  \cP = i_{\leq p}, \bl = \ell) . 
\end{align*}
\end{claim}



\begin{proof}


The proof is by induction on the number of rounds of $\Pi_{i_\ell}^{A \leftrightarrow B} := \Pi^{1}_{i_\ell} , \ldots , \Pi_{i_\ell}^{C}$. 
If at $\tau \in [C]$, it is Alice's turn to speak, then since Alice's message is a function of $S_{i_\ell} \Pi_{i_{\ell}}^{\cM}$ and $\Pi_{i_\ell}^{< \tau}$, it holds that 
\begin{equation} \label{eq:oddround}
    I(\Pi_{i_\ell}^\tau ; UT | S_{i_\ell} \Pi_{i_{\ell}}^{\cM} S_{i_{< \ell}} \Pi_{i_{< \ell}}^{\cM}, \Pi_{i_\ell}^{< \tau},  
    \cP = i_{\leq p}, \bl = \ell )  \\ 
    \leq H(\Pi_{i_\ell}^\tau |  S_{i_\ell}, \Pi_{i_{\ell}}^{\cM}, \Pi_{i_\ell}^{< \tau}) 
    = 0
\end{equation}
(Note that this would not have been true had Alice's message been a function of all $\vec{S}$, because $U$ correlates $\vec{S}$ and $T$. 
This is where we use the fact that only Megan's message $\Pi_{i_{< \ell}}^{\cM}$ can depend on all $\vec{S}$). 
If it is Bob's turn to speak at round $\tau \in [C]$, then it still holds that 
\begin{equation} \label{eq:evenround}
    I(\Pi_{i_\ell}^\tau ; S_{i_\ell} | UT, S_{i_{< \ell}} \Pi_{i_{< \ell}}^{\cM}, \Pi_{i_{\ell}}^{\cM}, \Pi_{i_\ell}^{< \tau}, 
    \cP = i_{\leq p}, \bl = \ell) = 0
\end{equation}
since Bob's message is determined by $UT, \Pi_{i_{\ell}}^{\cM}$ and $\Pi_{i_\ell}^{< \tau}$ or equivalently 
\begin{align*}
	& I(\Pi_{i_\ell}^\tau ; S_{i_\ell}  | UT, S_{i_{< \ell}} \Pi_{i_{< \ell}}^{\cM},\Pi_{i_{\ell}}^{\cM}, \Pi_{i_\ell}^{< \tau}, \cP = i_{\leq p}, \bl = \ell ) \\
	& \leq H(\Pi_{i_\ell}^\tau |  UT, \Pi_{i_{\ell}}^{\cM}, \Pi_{i_\ell}^{< \tau} ) = 0.
\end{align*}
Then applying \pref{fact:chainrule2} iteratively with \pref{eq:oddround} and \pref{eq:evenround} for any $\ell \in [p]$, we get
\begin{align*}
    & I(S_{i_{\ell}} ; UT | \Pi_{i_{\ell}}^{\cM} \Pi_{i_\ell}^{\leq C} S_{i_{< \ell}} \Pi_{i_{< \ell}}^{\cM}, \cP = i_{\leq p}, \bl = \ell) \\
    & \leq I(S_{i_{\ell}} ; UT | \Pi_{i_{\ell}}^{\cM} \Pi_{i_\ell}^{< C } S_{i_{< \ell}} \Pi_{i_{< \ell}}^{\cM}, \cP = i_{\leq p}, \bl = \ell) \leq \ldots \\
    &  \leq I(S_{i_\ell}  ; UT | \Pi_{i_{\ell}}^{\cM} \Pi_{i_\ell}^{1 } S_{i_{< \ell}} \Pi_{i_{< \ell}}^{\cM},  \cP = i_{\leq p}, \bl = \ell) \\
    &  \leq I(S_{i_\ell}  ; UT | \Pi_{i_{\ell}}^{\cM}S_{i_{< \ell}} \Pi_{i_{< \ell}}^{\cM},  \cP = i_{\leq p}, \bl = \ell).
\end{align*}
We get the final inequality by non-negativity of mutual information or
\begin{equation*}
I(S_{i_\ell}  ; UT | \Pi_{i_{\ell}}^{\cM}S_{i_{< \ell}} \Pi_{i_{< \ell}}^{\cM},  \cP = i_{\leq p}, \bl = \ell) \leq I( S_{i_\ell} \Pi_{i_{\ell}}^{\cM} ; UT | S_{i_{< \ell}} \Pi_{i_{< \ell}}^{\cM},  \cP = i_{\leq p}, \bl = \ell).
\end{equation*}

\end{proof}

We are finally ready to prove \pref{lem:goodq} using \pref{cl:pathsampling}, \pref{cl:smallcorr} and \pref{cl:roundelimination}.

\begin{proofof}{\pref{lem:goodq}}
Recall the definition of 
$Z^{\DISJ}(S_{I_{\bl}},T):= \Pi_{I_{\bl}} \Pi_{I_{< \bl}}^{\cM} ,\cP, \bl$.  
The correctness guarantee of $\DISJ_n(S_{I_{\bl}},T)$ holds since we set $\Pi_{I_{\bl}}^{ans}$ as $Z^{\DISJ}_{ans}$, and the original NOF protocol $\Gamma = (U,\Pi)$ was assumed to have 0 error. 

To establish \pref{eq:smallS}, we get from \pref{cl:pathsampling} and \pref{fact:chainrule} that 
\begin{align*}
    I(\Pi_{I_{\bl}} S_{I_{< \bl}} \Pi_{I_{< \bl}}^{\cM} T ,\cP, \bl ; S_{I_{\bl}} ) & = I(S_{I_{< \bl}} \Pi_{I_{< \bl}}^{\cM} T ,\cP, \bl  ; S_{I_{\bl}} ) + I( \Pi_{i_\ell} ; S_{i_\ell} | S_{I_{< \ell}} \Pi_{I_{< \ell}}^{\cM} T ,\cP, \bl ) \\
    & \leq  \frac{p \cdot C }{k - p} + C.
\end{align*}
Since we set $p := o ( k )$, we get $\frac{p \cdot C }{k - p} = o( C )$.

Next for \pref{eq:smallT}, we write
\begin{align}
    I(Z^\DISJ; T ) = \; 
    & I(\Pi_{I_{\bl}} S_{I_{< \bl}} \Pi_{I_{< \bl}}^{\cM},\cP, \bl ; T ) = \underbrace{I(\cP, \bl ; T)}_{=0} + I(S_{i_{< \ell}} \Pi_{i_{< \ell}}^{\cM} ; T | \cP, \bl ) + \underbrace{I(\Pi_{i_\ell} ; T | S_{i_{< \ell}} \Pi_{i_{< \ell}}^{\cM} , \cP, \bl )}_{\leq C } \nonumber \\
    & \leq \underbrace{I(S_{i_{< \ell}} \Pi_{i_{< \ell}}^{\cM} ; T | \cP, \bl )}_{=0} + C \leq C. \label{eq:point_of_deviation}
\end{align}
where we used \pref{eq:non-adaptive} for $I(S_{i_{< \ell}} \Pi_{i_{< \ell}}^{\cM} ; T | \cP, \bl ) = I(S_{i_{< \ell}} \Pi_{i_{< \ell}}^{\cM} ; T ) = 0$ (recall that $\cP, \bl$ are 
chosen independently of the inputs, so conditioning on them does not change things). 
We remark that here we (crucially) used the fact that Megan is only allowed to send a single 
message (i.e., no further interaction with the player holding $\vec{S}$ is allowed). 

Finally for \pref{eq:smallCorrelation}, from \pref{cl:roundelimination}, we have that for every $\ell  \in [p]$ and $\cP$,
\begin{align} 
    I( S_{i_{\ell}} ; T | \Pi_{i_{\ell} }, S_{i_{< \ell}}, \Pi_{i_{< \ell }}^{\cM} , \bl, \cP ) \leq I( S_{i_\ell} \Pi_{i_{\ell}}^{\cM} ; U T | S_{i_{< \ell}}, \Pi_{i_{< \ell}}^{\cM} , \bl, \cP). \label{eq:applyelimination}
\end{align}
But from \pref{cl:smallcorr}, we know that over random $\bl \in_R [p]$, we have
\begin{align} 
    & I( S_{i_\ell} \Pi_{i_{\ell}}^{\cM}  ; U T | S_{i_{< \ell}} \Pi_{i_{< \ell}}^{\cM}, \bl, \cP ) \nonumber \\
    &= \E_{ \bl, \cP } \left[ I( S_{i_\ell} \Pi_{i_{\ell}}^{\cM}  ; U T | S_{i_{< \ell}} \Pi_{i_{< \ell}}^{\cM}, \cP ) \right] \leq O \left( \frac{|UT|}{p} \right). \label{eq:smallcorrelation} 
\end{align}
\end{proofof}

\subsection{A Random Process for AND with Low Information \texorpdfstring{$\&$}{and} Correlation} 
We now show how to ``scale down" the random process $Z^{DISJ}$ (obtained from the NOF protocol $\Gamma$ in \pref{lem:goodq}) 
so as to generate another random process (\emph{not} a 2-party protocol) 
that ``approximately" computes the $2$-bit $\AND(X,Y)$ function (on independent random bits) under $X,Y \sim \cB_\gamma$, 
with information correlation smaller by a factor of $n$. This follows the standard ``direct sum" embedding (e.g., \cite{BBCR10, Braverman2011, Bra15})  
-- the important observation here is that the direct sum property of the information cost function apply to general interactive processes 
and not just to communication protocols. This is the content of the next lemma. 


\begin{lemma} \label{lem:reduction}
Let $\Gamma = ( U, \Pi )$ be a 4-party NOF protocol that solves $\SEL{\DISJ_n}{k}$ with $|U| < o(p) < o(k)$ and $|\Pi_i| < C$ 
for all $i \in [k]$. Then there exists a random variable $Z^{\AND} = Z^{\AND}(X,Y)$ such that 
\begin{itemize}
    \item If $\AND(X, Y) =1$, then $Z^{\AND}_{ans}$ outputs $0$.
    \item If $\AND(X, Y) = 0$, then $Z^{\AND}_{ans}$ outputs $0$ with probability at most $0.001$.
    \item Has following information cost guarantees
    \begin{align}
    & I( Z^{\AND}; X ) \leq \frac{C + o( C )}{n} \label{eq:smallx} \\
    & I( Z^{\AND} ; Y ) \leq \frac{C}{n} \label{eq:smally} 
    \end{align}
    \item Has following correlation guarantee
    \begin{align}
    & I( X ; Y | Z^{\AND} ) < o ( 1 / n ) \label{eq:singlecorr}
    \end{align}
\end{itemize}
\end{lemma}


\begin{proof}
	Consider the following embedding of bits $X$ and $Y$ to $\Gamma$.
	\begin{Protocol}
	\begin{enumerate}
	\item Select $\cP$, $\bl \in_R [p]$, $\bj \in_R [n]$ uniformly at random.
	\item Set $S_{I_{\bl}}^{\bj} = X$ and $T^{\bj} = Y$.
	\item Sample the rest of the coordinates all i.i.d. $\cB_\gamma$.
	\item Run $\Gamma = (U,\Pi)$ with $I_{\bl}$ as the index. Then return the output.
	\end{enumerate}
	\caption{Embedding $\AND(X,Y)$ \label{prot:embedding}}
	\end{Protocol}
	
	Now we claim that $$Z^{\AND} := \Pi_{I_{\bl}} S_{I_{< \bl}} \Pi_{I_{< \bl}}^{\cM} T^{< \bj} \bj \cP \bl = Z^{\DISJ} T^{ < \bj } \bj $$ 
	with $Z^{\AND}_{ans}$ set as $\Pi_{I_{\bl}}^{ans}$ satisfies the conditions of \pref{lem:reduction}.
	
	\paragraph{Correctness} Suppose $\AND(X,Y) = 1$. Then $\Pi_{I_{\bl}}$ must output 0 since $\DISJ(S_{I_{\bl}},T) = 0$ from embedding. Also note that given $\AND(X,Y) = 0$, $\DISJ(S_{I_{\bl}},T) = 0$ with at most $0.001$ probability since $\DISJ(S_{I_{\bl}}^{- \bj}, T^{- \bj}) = 0$ with at most $0.001$ probability on our distribution on $\vec{S}$ and $T$ for any setting of $\bj, \cP, \bl$. Therefore $\Pi_{I_{\bl}}$ outputs $0$ with at most probability $0.001$.
	
	\paragraph{Information Cost}
	Recall that we have from \pref{lem:goodq},  \pref{eq:smallS}, \pref{eq:smallT}, \pref{eq:smallCorrelation} or
	\begin{align*}
    & I(Z^{\DISJ} T ; S_{I_{\bl}} ) = I(  \Pi_{I_{\bl}} S_{I_{< \bl}} \Pi_{I_{< \bl}}^{\cM} T, \cP, \bl ; S_{I_{\bl}} ) \leq C+ o( C ) \\
    & I(Z^{\DISJ}  ; T ) = I( \Pi_{I_{\bl}} S_{I_{< \bl}} \Pi_{I_{< \bl}}^{\cM} T, \cP, \bl ; T ) \leq C  \\
    & I( S_{i_\ell}  ; T | Z^{\DISJ} ) = I( S_{i_\ell} ; T | \Pi_{i_{\ell}} S_{i_{< \ell}} \Pi_{i_{< \ell}}^{\cM} T, \cP, \bl) \leq O \left( \frac{|UT|}{p} \right) 
    \end{align*}
    Now since we embed $X$ and $Y$ in random $\bj \in [n]$, we get
    \begin{align*}
    	& I ( \Pi_{I_{\bl}} S_{I_{< \bl}} \Pi_{I_{< \bl}}^{\cM} T^{< \bj }, \bj, \cP, \bl ; X) 
    	= \E_{\cP, \ell} \E_{j} \left[ I ( \Pi_{i_\ell} S_{i_{< \ell}} \Pi_{i_{< \ell}}^{\cM} T^{< j}  ; X | \bj = j, \cP, \bl) \right] \\
    	& \leq \E_{\cP, \ell} \E_{j} \left[ I ( \Pi_{i_\ell} S_{i_{< \ell}} \Pi_{i_{< \ell}}^{\cM} T  ; S^j_{i_\ell} | \bj = j, \cP, \bl ) \right]
    	= \E_{\cP, \ell} \E_{j} \left[ I ( \Pi_{i_\ell} S_{i_{< \ell}} \Pi_{i_{< \ell}}^{\cM} T  ; S^j_{i_\ell} | \cP, \bl) \right] \\
    	 &=  \E_{\cP, \ell} \left[ \frac{1}{n}  \sum_{j \in [n]} I ( \Pi_{i_\ell} S_{i_{< \ell}} \Pi_{i_{< \ell}}^{\cM} T ; S^j_{i_\ell} | \cP, \bl) \right]
        \leq \E_{\cP, \ell} \left[ \frac{1}{n} \sum_{j \in [n]} I ( \Pi_{i_\ell} S_{i_{< \ell}} \Pi_{i_{< \ell}}^{\cM} T ; S^j_{i_\ell} | S^{< j}_{i_\ell}, \cP, \bl) \right] \\
        & \leq \E_{\cP, \ell} \left[ \frac{I ( \Pi_{i_\ell} S_{i_{< \ell}} \Pi_{i_{< \ell}}^{\cM} T ; S_{i_\ell} | \cP, \bl)}{n} \right] = \frac{I ( \Pi_{I_{\bl}} S_{I_{< \bl}} \Pi_{I_{< \bl}}^{\cM} T, \cP, \bl ; S_{I_{\bl}} )}{n} \\
        & = \frac{I(Z^{\DISJ} T; S_{I_{\bl}} )}{n} \leq \frac{C + o(C)}{n}
    \end{align*}
    where the second inequality holds from $I(S^j_{i_\ell};S^{<j}_{i_\ell}|\cP, \bl) =0$. The second to last equality holds from $I(\cP, \bl; S_{I_{\bl}}) = 0$ or \pref{eq:nopl}.
    For \pref{eq:smally}, analogously, we get 
    \begin{align*}
    	& I ( \Pi_{I_{\bl}} S_{I_{< \bl}} \Pi_{I_{< \bl}}^{\cM} T^{< \bj }, \bj, \cP, \bl ; Y) 
    	= \E_{\cP, \ell} \E_{j} \left[ I ( \Pi_{i_\ell} S_{i_{< \ell}} \Pi_{i_{< \ell}}^{\cM} T^{< j}  ; T^j | \bj = j, \cP, \bl) \right] \\
    	& \leq \E_{\cP, \ell} \E_{j} \left[ I ( \Pi_{i_\ell} S_{i_{< \ell}} \Pi_{i_{< \ell}}^{\cM} T^{<j}  ; T^j | \cP, \bl) \right]
    	= \E_{\cP, \ell} \E_{j} \left[ I ( \Pi_{i_\ell} S_{i_{< \ell}} \Pi_{i_{< \ell}}^{\cM} ; T^j| T^{<j}, \cP, \bl) \right] \\
    	&=  \E_{\cP, \ell} \left[ \frac{1}{n} \sum_{j \in [n]} I ( \Pi_{i_\ell} S_{i_{< \ell}} \Pi_{i_{< \ell}}^{\cM} ; T^j | T^{< j},\cP, \bl ) \right]
        \leq \E_{\cP, \ell} \left[ \frac{I ( \Pi_{i_\ell} S_{i_{< \ell}} \Pi_{i_{< \ell}}^{\cM}  ; T |\cP, \bl)}{n} \right] \\
        & = \frac{I ( \Pi_{I_{\bl}} S_{I_{< \bl}} \Pi_{I_{< \bl}}^{\cM} , \cP, \bl ; T )}{n}  = \frac{I(Z^{\DISJ}; T )}{n} \leq \frac{C}{n}.
    \end{align*}
    where the second equality holds from $I(T^{<j} ; T^j | \cP, \bl) = I(T^{<j} ; T^j ) = 0$.
    
    \paragraph{Low Correlation}
    Now for \pref{eq:singlecorr}, we again take expectation over $\bj$.
    \begin{align*}
    	& I(X;Y| \Pi_{i_{\ell}} S_{i_{< \ell}} \Pi_{i_{< \ell}}^{\cM} T^{< j }, \bj, \cP, \bl
    ) = \E_{\cP, \ell} \E_{j} \left[ I ( S_{i_\ell}^j; T^j | \Pi_{i_{\ell}} S_{i_{< \ell}} \Pi_{i_{< \ell}}^{\cM} T^{< j }, \cP, \bl, \bj = j ) \right] \\
    & = \E_{\cP, \ell} \E_{j} \left[ I ( S_{i_\ell}^j; T^j | \Pi_{i_\ell} S_{i_{< \ell}} \Pi_{i_{< \ell}}^{\cM} T^{<j}, \cP, \bl ) \right] \\
     & \leq \E_{\cP, \ell} \left[ \frac{1}{n} \sum_{j} I ( S_{i_\ell}; T^j | \Pi_{i_\ell} S_{i_{< \ell}} \Pi_{i_{< \ell}}^{\cM} T^{<j}, \cP, \bl ) \right]\\
     & = \E_{\cP, \ell} \left[ \frac{I ( S_{i_\ell}; T | \Pi_{i_\ell} S_{i_{< \ell}} \Pi_{i_{< \ell}}^{\cM} \cP, \bl)}{n} \right] = \frac{I(S_{I_{\bl}}; T | Z^{\DISJ} )}{n} \leq O \left( \frac{|UT|}{pn} \right) \leq o(1/n).
    \end{align*}
    where the last bound follows from \pref{lem:goodq} with $|UT| < o( p )$.
    \end{proof}

\subsection{An Information Lower Bound for Low-Correlation AND Computation}
In this section, we rule out a random process that computes the 2-bit $\AND$ function with {\em simultaneously} low information cost and small correlation. The proof can be viewed as a generalization of the classic Cut-and-Paste Lemma \cite{CutPaste} from 2-party protocols to a more general setting of low-correlation random variables.

\paragraph{A Robust ``Cut-and-Paste" Lemma for Product Distribution.} 
Instead of restricting to a two-party protocol, we generalize the setting for ``cut-and-paste" to any random process. We prove the following general structural bound for any random process $Z$ when $X$ and $Y$ are distributed $B_\gamma$ independently with $\gamma = o(1)$.

\begin{lemma}[Robust Cut and Paste] \label{lem:andbound}
Suppose $X$ and $Y$ are distributed i.i.d. $\cB_\gamma$ with $\gamma = o(1)$. Then consider a random variable $Z$ containing $Z_{ans}$ such that
\begin{itemize}
\item If $\AND(X,Y) = 1$ then $Z_{ans} = 0$. Otherwise $\Pr[ Z_{ans} = 0 ] < 0.001$.
\item $\Pr[ Z_{ans} = 1 ] \geq 1/2$
\item Satisfies following two inequalities
\begin{align}
	& I(X ; Y | Z ) \leq o ( \gamma^2 ) \label{eq:correlation} \\
	& I( Z  ; X  ) \leq o( \gamma )
\end{align}
\end{itemize}
Then it must be the case $I( Y ; Z  ) \geq \Omega ( \gamma )$. 
\end{lemma}

We remark that in the usual cut-and-paste setting introduced in \cite{CutPaste}, one requires $I(X;Y|Z) = I(X;Y)$. And this follows from $Z$ being a communication protocol. It is crucial in standard cut-and-paste argument that $I(X;Y|Z) = I(X;Y)$, and therefore for communication model in \pref{fig:model} which crucially introduces correlation between two inputs conditioned on the protocol, using cut-and-paste seemed far from giving any bound. Instead, in \pref{lem:andbound}, we show that if $X$ and $Y$ are under product distribution, one can actually relax this condition and make the argument robust to small correlation between inputs (even if $Z$ is not a communication protocol). In particular $Z$ can be an arbitrary random process!

\begin{proofof}{\pref{lem:andbound}}
	Towards the proof we make use of the following notation for distribution on $X, Y$ conditioned on $Z = z$ assuming $Z_{ans} = 1$. 
\begin{align*}
a_{z} := \Pr[ X = 0 , Y = 0 | Z = z] \\
b_{z} := \Pr[ X = 0 , Y = 1 | Z = z] \\
c_{z} := \Pr[ X = 1 , Y = 0 | Z = z] 
\end{align*}
Note that  $\Pr[ X  = 1 , Y = 1 | Z = z] = 0 $, since we have the guarantee that if $\AND(X,Y) = 1$ then $Z_{ans} = 0$. We prove the following claim on $a_z,b_z,c_z$.

\begin{figure}[!ht]
\centering
\includegraphics[scale=0.25]{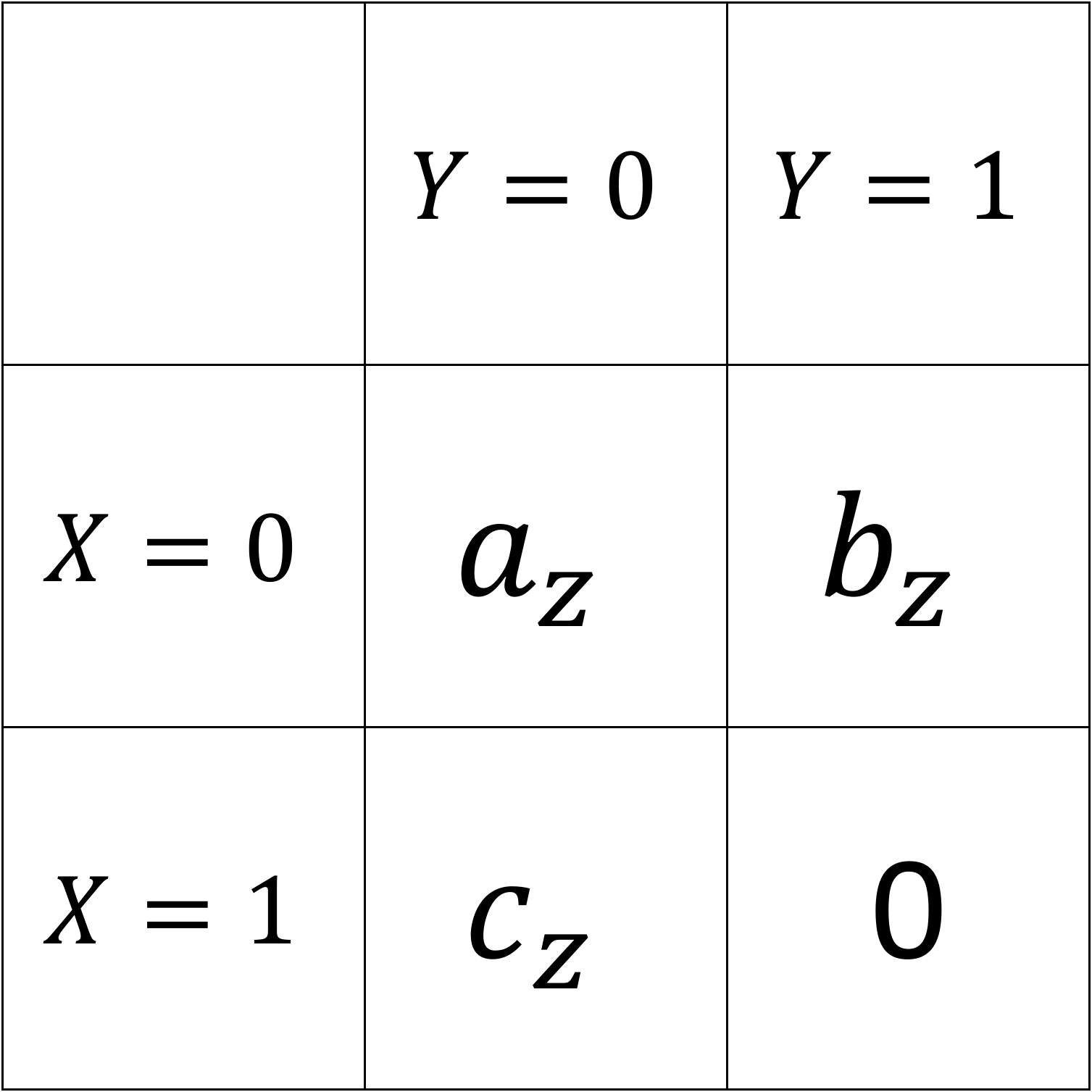}
\caption{Resulting prior conditioned on $Z = z$ with $Z_{ans} =1$}
\label{fig:xy_prior}
\end{figure}

\begin{claim} \label{cl:largediv}
	If $\KL{X_{z}}{X} \leq o(\gamma)$,
	and $I(X; Y | Z = z) \leq o ( \gamma^2 )$, 
	then $\KL{Y_{z}}{Y} \geq \Omega( \gamma ).$
\end{claim}
\begin{proof}
    First, observe that $c_{z} = \Theta( \gamma )$. If $\KL{X_{z}}{X} \leq o (\gamma)$, then from \pref{fact:divergence} we have 
    \begin{equation}\label{eq:smallc}
    \Pr[X = 1 | Z = z] = c_{z} = \Theta( \gamma ).
    \end{equation}
	Next, we expand and lower bound the term $I(X;Y|Z = z )$. 
	\begin{align*}
	& I(X;Y| Z = z) = \Pr[ Y = 0 | Z = z ] \cdot \KL{ X_{Y=0, z}}{ X_{z}} \\
	& + \Pr[ Y = 1 | Z = z ] \cdot \KL{ X_{Y=1, z}}{ X_{z}} \\
	& \geq ( 1 - b_{z} ) \cdot \KL{\cB_{\frac{c_{z}}{1 - b_{z} }}}{ \cB_{c_{z}}} + b_{z} \cdot \KL{\cB_{0}}{\cB_{c_{z}}} \\
	& \geq  b_{z} \cdot \KL{\cB_{0}}{ \cB_{c_{z}} } = b_{z} \log \frac{1}{1-c_z} \geq \Omega \left( b_{z}  c_{z}   \right).
	\end{align*}
	where the last bound holds from $- \log ( 1 - x ) \geq x / 2 $ for $x < 1/2$.
	Rewriting the inequality we get,
	\begin{equation} \label{eq:bc}
	b_{z} c_{z} \leq O \left( I(X;Y| Z = z ) \right)
	\end{equation}
    Now we have that $c_{z} = \Theta( \gamma )$ from \pref{eq:smallc}. Therefore we can rewrite \pref{eq:bc} as 
	\begin{equation} \label{eq:c-bound}
	b_{z} \leq O \left( \frac{I(X;Y|Z = z )}{\gamma} \right)
	\end{equation}
	Since we assumed $I(X;Y| Z = z) \leq o ( \gamma^2 )$, we obtain
	\begin{equation} \label{eq:b-bound}
	b_{z} \leq o \left( \gamma \right)
	\end{equation}
	Then combining \pref{eq:b-bound} with \pref{fact:divergence}, we get 
	\begin{equation*}
	    \KL{Y_{z}}{Y} \geq \Omega ( \gamma ).
	\end{equation*}
\end{proof}
\noindent To complete the proof of \pref{lem:andbound}, take ``good" $z$ such that 
\begin{itemize}
    \item $z_{ans} = 1$ 
    \item $\KL{X_{z}}{X} \leq o ( \gamma )$
    \item $I(X;Y | Z = z ) \leq o( \gamma^2 )$
\end{itemize}
By union bound, the mass on such $Z$ must be at least $\Omega(1)$. Furthermore, for these $z$, \pref{cl:largediv} holds, and $\KL{Y_{z}}{Y} \geq \Omega( \gamma )$. Then 
\begin{align*}
    & I( Z ; Y ) = \E_{z \sim Z } \left[ \KL{Y_{z}}{Y} \right] \geq \Pr[ z \mbox{ is good} ] \cdot \Omega( \gamma ) = \Omega(\gamma)
\end{align*}
where the expectation sum is only taken over good $z$.
\end{proofof}

\subsection{Proof of the Main Theorem}
The proof of Theorem \ref{thm:4party} now follows easily by combining \pref{lem:reduction} and \pref{lem:andbound}.  
\begin{theorem}[Restated] \label{thm:mainLB}
Let $\Gamma = (U, \Pi)$ be a 4-party NOF  protocol (c.f. \pref{fig:4party}) that solves $\SEL{\DISJ_n}{k}$ with $|T| < |U| < o( k )$. Then there exists some $i \in [k]$ such that $|\Pi_i| \geq \Omega(\sqrt{n})$.
\end{theorem}
\begin{proof}
	Suppose otherwise. Then for all $i \in [k]$, we have $| \Pi_i | < o (\sqrt{n})$. Then by \pref{lem:reduction}, there is some random variable 
	$Z^{\AND}(X,Y)$ for solving $\AND(X,Y)$ with the following guarantees by setting $C = o( \sqrt{n} )$ : 
	\begin{align*}
    	& I(Z^{\AND} ; X ) < o ( 1 / \sqrt{n} ) \\
    	& I(Z^{\AND} ; Y ) < o ( 1 / \sqrt{n} ) \\
    	& I( X ; Y | Z^{\AND} ) < o ( 1 / n ).
	\end{align*}
	Furthermore, $Z^{\AND}_{ans} = 0$ whenever $\AND(X,Y) = 1$. But $\Pr[\AND(X,Y) = 1] = O(1/n)$. While if $\AND(X,Y) = 0$, $\Pr[Z^{\AND}_{ans} = 0] < 0.001$. Therefore, $\Pr[Z^{\AND}_{ans} 	= 1] \geq 1/2$. But this is in direct contradiction to \pref{lem:andbound} with $\gamma := \frac{1}{1000 \sqrt{n}}$. Hence such $Z^{\AND}$ cannot exist.
	
\end{proof}


\subsection{Multiphase Lower Bound for Semi-Adaptive Data Strucutres}

Here we prove Theorem \ref{thm_dynamic_multiphase_informal}, asserting a polynomial lower bound on the Multiphase problem against 
Semi-Adaptive dynamic data structures (\pref{def:semiadaptive}): 

\begin{theorem}[Main Result]
	Let $k = \omega( n )$. Any semi-adaptive  data structure solving the Multiphase problem 
	must have (total) update time $t_u n \geq \Omega( \frac{k}{ w } )$ or query time $t_q \geq \Omega( \sqrt{n} / w )$, 
	in the cell-probe model with word size $w$. 
\end{theorem}

\begin{proof}
	We use a simple variation of the reduction from \cite{patrascu:multiphase} 
	to show that an efficient semi-adaptive data strucutre implies a too-good-to-be-true 4-party NOF protocol, contradicting 
	\pref{thm:mainLB}. To this end, 
	suppose we have a semi-adaptive dynamic data structure with $t_u n < o \left( \frac{k}{ w } \right)$ and $t_q < o( \sqrt{n} / w )$.
	
	We argue that this implies a cheap 4-party NOF protocol. First we set the update transcript as $U$ -- which can be generated by Charlie who has access to both $\vec{S}$ and $T$. We then have $|U| = O(t_u n w) < o( k )$.
	Set $\Pi^{\cM}_i$ as addresses and the contents of cells in $\cM$ that are accessed by the query algorithm. 
	Set subsequent Alice's message $\Pi_i^\tau$ as the cell address of $\Delta(\cM,T)$ that are accessed. This is determined by $\Pi_i^{\cM}, \Pi_i^{<\tau}, i, S_i$. Set Bob's message $\Pi_i^{\tau+1}$ as the cell content. This is determined by $\Pi_i^{\cM}, \Pi_i^{\leq \tau}, UT, i$.
	It is immediate that this is a valid 4-party protocol. 
	We also have the length of 4-party NOF protocol $\Gamma = (U, \Pi)$ as
	\begin{align*}
	& |U| + |T| \leq o(k) \\
	& |\Pi_i| = |\Pi^{\cM}_i| + |\Pi^{A \leftrightarrow B}_i| \leq 4 t_q w < o ( \sqrt{n} )
	\end{align*}
	But this is in contradiction to \pref{thm:mainLB}.
\end{proof}


\section{Implications of the NOF Multiphase Conjecture} \label{sec:consequences}

In this section, we show that \Patrascu's NOF Conjecture on the original Multiphase Game, 
even against 3-round protocols, would imply a breakthrough in circuit complexity. 
This complements our restricted NOF model, as it shows that allowing even \emph{two} of Alice's 
messages to depend arbitrarily on her entire input $\vec{S},i$, would resolve a decades-old open problem 
in circuit lower bounds. It also suggests that attacking the Multiphase conjecture for (general) dynamic data strucutres via 
the the NOF Game, should exploit the fact that data strucutres induce highly restricted NOF protocols. 

First, note that Conjecture \ref{MPH_NOF_conj} in particular implies the following special case: 

\begin{conj}[3-round NOF Conjecture] \label{conj:NOF}
Any 3-round NOF protocol for the 3-party  Multiphase Game with $|U| = o(k)$ bits of advice must have $|\Pi| > n^\eps$ communication. 
\end{conj}

\paragraph{Circuits with arbitrary gates}
As mentioned in the introduction, 
a long-standing open problem in circuit complexity is whether \emph{non-linear} gates can significantly (polynomially) reduce the number of wires of circuits 
computing linear operators \cite{JS10}. 
We consider Valiant's depth-2 circuit model \cite{Val77} with arbitrary gates, and its generalizations to arbitrary depths. 
More formally, consider a circuit computing a linear operator $x\mapsto Ax$ where $A$ is an $k \times n$ matrix with $k =\tilde{O}(n)$, 
using unbounded fan-in, and where gates are allowed to be \emph{arbitrary} functions. 
Clearly, such circuits can trivially compute any $f$ with $k$ gates. As such, the interesting complexity measure in this model is the 
minimum number of {\bf wires} to computing the function $f$. This measure captures how much ``information" needs to be transferred 
between different components of the circuit, in order to compute the function. For a more thorough exposition and motivation on 
circuits with arbitrary gates, we refer the reader to \cite{Jukna_2012}.

\paragraph{Previous Works}

In contrast to arithmetic circuit models (e.g., \cite{Val77} where allowed functions are simple functions such as AND, OR or PARITY), it is a long-standing open problem \cite{JS10, Jukna_2012, Drucker12} whether non-linear circuits can compute {\em any} linear operator $A$ with near-linear ($\tilde{O}(k)$) number of wires. Indeed, for linear circuits, this is a simple counting argument. Counting argument shows that $\Omega( k^2 / \log k)$ wires are necessary for linear circuits, and this is tight \cite{lupanov56}. But once again, for arbitrary circuits, counting argument completely fails, since the number of possible functions over $n$ bits is already doubly exponential in $n$. 


\cite{jukna10} initiated works on analyzing the complexity of {\em representing} a random matrix, that is computing $A x$ when $x$ is restricted to having only one 1. In other words, compute $A e_i$ for $i \in [n]$. \cite{Drucker12} showed that when restricted to {\em representing} a matrix, $\Omega(k \log k)$ is necessary for depth $2$ circuit, complementing the previous upper bound of $O( k \log k)$ by \cite{jukna10}. 


The lower bound of \cite{Drucker12} immediately implies $\Omega(k \log k)$ lower bound for computing a matrix using depth $2$ circuit, since any circuit that computes a matrix must represent a matrix as well. But in case of computing the linear operator $A$, no better bounds were known. 
Furthermore, for higher depth circuits, no super-linear bounds are known, even for {\em representing} a matrix. In fact, \cite{Drucker12} showed that one cannot hope to obtain a better lower bound when restricted to {\em representing} a matrix by giving $O(k)$-sized depth-$3$ circuit for representing a matrix.

\

We first show that \pref{conj:NOF} implies a \emph{static} data structure lower bound (on a problem defined by a random 
but hard-wired matrix $A$).  We then use a recent reduction of \cite{V18} to show that this static data structure lower bound implies 
a  lower bound on the wire complexity of depth-$d$ circuits with arbitrary gates computing $Ax$. 

\paragraph{Static Data Structure Lower Bound} First, we consider the following class of static data structure problems $\cP^f_A (x)$ defined by a query matrix $A \in \{ 0 , 1 \}^{k \times n}$ and a 
function $f: \{ 0 , 1 \}^{2n} \rightarrow \{ 0 , 1 \}$: 

\begin{enumerate}
\item Given a \emph{fixed} matrix $A$ with rows $A_1, \ldots, A_k$, preprocess an input database $x \in \{ 0, 1 \}^n$. 
\item Given $i \in [k]$ as a query, the data structure needs to output $f (A_i,x)$. 
\end{enumerate}

Note that $A$ is \emph{hard-wired} to the problem, i.e., the data structure can access $A$ for free 
during both preprocessing and query stage\footnote{This is a generalization of Valiant's model \cite{Val77} in that the 
circuit itself is allowed to depend arbitrarily on $A$.}. In particular, with $s=k/w$ space, one can store the (boolean) 
answers for all possible queries and the problem becomes trivial ($t=1$), whereas without any preprocessing, the query 
algorithm needs to read $x$ (but not $A$) to compute the answer $f(A_i,x)$, giving (worst case) query time $t\sim n/w$. 
Accordingly, the query algorithm is \emph{non-adaptive} if the cell addresses that are probed only a function of 
$i \in [k]$ and $A$.

We show that \pref{conj:NOF} implies the following lower bound on $\cP^f_A (x)$ where $f := \DISJ_n$. 

\begin{lemma}[Polynomial Static Lower Bound for Random Set Disjointness Queries] \label{lem:static}
Suppose \pref{conj:NOF} holds. Let $k = \omega( n)$. Then there exists a collection of $k$ sets $A := A_1,\ldots, A_k \subseteq [n]^k$ such that 
any non-adaptive static data structure solving $\cP^{\DISJ_n}_A$ must either use $s \geq \Omega( \frac{k}{ w } )$ 
space, or have $t \geq \Omega( n^{\eps} / w )$ query time, in the cell-probe model with word size $w$. 
\end{lemma}
\begin{proof}
	Suppose for any $A \in \{0,1\}^{k\times n}$ there is a non-adaptive static data structure $D_A$ computing $\cP^{\DISJ_n}_A$ 
	with $s \leq o ( \frac{k}{ w } )$ space and $t \leq o( n^{\eps} / w )$ cell probes. We show that this induces a too-good-to-be-true 
	(3-round) NOF protocol for the Multiphase game $\SEL{\DISJ_n}{k}$, violating \pref{conj:NOF}. Indeed, consider the 
	following simple 3-party protocol for simulating $D_A$: 
	
	Charlie's ``advice" $U$ in Phase 1 of the Multiphase game will be the contents of the $s$ memory cells of $D_A$, 
	Alice's message during Phase 2 (i.e., $\Pi_{i}^{A \rightarrow B}$) will be the memory addresses probed by the $D_A$ for 
	answering $\DISJ(A_i, x)$, and Bob's messages $\Pi_{i}^{B \rightarrow A}$ are the contents of cells probed by Alice. 
	Note this protocol is well defined: Indeed, by the definition of $\cP^{\DISJ_n}_A$, $U$ only depends on $A$ and $x$ 
	at preprocessing time, and if $D_A$ is non-adaptive, then $\Pi_{i}^{A \rightarrow B}$ is only a function of $A$ and $i$; 
	Finally, $\Pi_{i}^{B \rightarrow A}$ depends on the previous transcript and $U= U(A,x)$ which Bob possesses. We therefore 
	have a valid 3-round NOF protocol for $\SEL{\DISJ_n}{k}$ with 
	\begin{align*}
	& |U| + |x| \leq sw + n \leq o( k ) + n = o( k ) , \text{\;\;\; and}  \\
	& |\Pi_i| = |\Pi_{i}^{A \rightarrow B}| + |\Pi_{i}^{B \rightarrow A}| \leq 2 t w \leq o( n^{\eps} )
	\end{align*}
	bits, which contradicts \pref{conj:NOF}.
\end{proof}



We consider the following parameter for the circuit lower bound.

\begin{corollary} \label{cor:forcircuitbound}
If $k = \omega(n)$, and $s = O( n / w )$ then $t \geq \Omega( n^{\eps} / w)$.	
\end{corollary}


\paragraph{Circuit Lower Bound} Now we show that \pref{conj:NOF} implies circuit lower bounds using reduction by \cite{V18} from \pref{lem:static}. We use the following translation theorem for lower bounds in arbitrary depths.



\begin{theorem}[\cite{V18}]
\label{thm:dstocircuit}
Suppose function $f: \{0,1\}^n \rightarrow \{ 0 , 1 \}^k$ has a circuit of depth $d$ with $\ell$ wires, consisting of unbounded fan-in, arbitrary gates. Then for any $r$ there exists a static data structure (with non-adaptive query) with space $s = n + r$, query time $(\ell/r)^d$, and word size $\max \{ \log n, \log r \} + 1$ which solves the following problem
\begin{enumerate}
\item Preprocess input $x$ depending on $x$ and $f$
\item Given $i \in [k]$, output $f_i(x)$.	
\end{enumerate}
\end{theorem}

For completeness, we attach the proof here. Though \cite{V18} did not remark on queries being non-adaptive, we remark that data structures derived from circuits are intrinsically non-adaptive, that is the cells probed only depends on query index and $f$.

\begin{proofof}{\pref{thm:dstocircuit}}
Consider circuit $C_f$ which computes $f$. Now set $G$ as the set of gates with in-wire $> \ell / r $. Since the number of wires are bounded by $w$, $|G| < r$. Now given $x$, store the values of these gates $G$ in space $r$.

Now we argue inductively on the level of the gate. We show that if the gate is at level $m$, that the number of non-adaptive queries made to compute the value of the gate is at most $( \ell /r)^m$. As base case, suppose if it were a level $1$ gate. If it lies in $G$, then the non-adaptive query required is 1, by calling to its address in $G$, which requires $\log |G|$-bits. Otherwise, the query required is at most $\ell / r$, each of which requires $\log n$-bits for the address, and they are non-adaptive. Therefore, the word size required is $\max \{ \log n, \log r \} + 1$.

Now as induction step, suppose for any $j < m$, level $j$ gates require at most $(\ell/r)^j$ non-adaptive queries to compute its value with word size $\max \{ \log n, \log r \} + 1$. Consider a level $m$ gate. If the gate lies in $G$, again it only requires 1 non-adaptive query, by calling to its address in $G$. Otherwise, it can be answered by computing at most $\ell/r$ level $m-1$ gates, and these queries are non-adaptive. Each of these level $m-1$ gates require $(\ell / r)^{m-1}$ non-adaptive queries by induction hypothesis, and the word size required is $\max \{ \log n, \log r \} + 1$. Therefore, the total non-adaptive query required is at most $(\ell / r) \cdot (\ell / r)^{m-1} = (\ell / r)^{m}$.

Since the final output gate that we are interested in is at level $d$, we get $(\ell / r)^{d}$ as the upper bound on the number of queries.
\end{proofof}

Here we crucially used the fact that wirings are fixed if we fix $C_f$. Note that $G$ is determined by $C_f$. Then a contrapositive of \pref{thm:dstocircuit} then states that data structure lower bound with non-adaptive query implies circuit lower bounds. When restricted to linear space usage (i.e. $r = O(n)$, $f(n)$ query time lower bound translates to $\Omega \left( n f(n)^{1/d} \right)$ lower bound for $\ell$. Using the contrapositive along with \pref{cor:forcircuitbound}, we show the following circuit lower bound.


\begin{corollary} \label{cor:circuit}
Assuming \pref{conj:NOF} and $k = \omega( n )$, there exists a matrix $A \in \{ 0 , 1 \}^{k \times n}$ such that any depth-$d$ circuit 
	that computes $A x$\footnote{under Boolean Semi-Ring, with addition as OR and multiplication as AND} requires $n^{1 + \Omega \left( \eps / d \right)}$ wirings. In particular, if $d = 2$, there exists $A$ that requires $n^{1 + \frac{\eps}{2} - o(1)}$ wirings.
\end{corollary}


\begin{proof}
We use contrapositive of \pref{thm:dstocircuit}, that is data structure lower bound implies circuit lower bounds. Consider \pref{cor:forcircuitbound}. Setting $r = O(n)$, and $k = \omega( n )$. Then \pref{cor:forcircuitbound} with \pref{thm:dstocircuit} implies that for some $A$, $(\ell/r)^d \geq \Omega ( n^{\eps} / \log n)$. Since $r = O(n)$, rewriting in terms of $\ell$, we get
\begin{equation*}
\ell \geq \Omega \left( n^{1+ \frac{\eps}{d} - \frac{\log \log n}{d \log n} } \right)	=  n^{1+ \Omega \left( \frac{\eps}{d} \right) }  .
\end{equation*}
Now setting $d = 2$, we get $\ell \geq  n^{ 1 + \frac{\eps}{2} - o(1)} $. 

\end{proof}


\section{Discussion} \label{sec:openproblem}

\paragraph{Extending our techniques to fully adaptive queries}
A natural question to ask is if our technical approach can be extended to fully adaptive queries as well,   
so as to resolve the Multiphase Conjecture. 
Our results indicate that attacking this problem through the NOF Multiphase Game should 
exploit the fact that (general) data structures induce \emph{restricted} NOF protocols -- namely, that 
the query algorithm only has limited \emph{local} access to its memory $\cM(\vec{S},T)$, through 
previously probed cells (unlike Alice in the NOF game, who has general access to $\vec{S}$).  

 
Let $\Pi_D$ be a NOF protocol induced by an efficient dynamic data structure $D$ for $\SEL{\DISJ_n}{k}$. 
Recall that the main technical challenge in our proof is to use $\Pi$ to design a random variable $Z(X,Y)$ computing 
$\AND(X,Y)$, while  \emph{simultaneously} controlling $I(Z;Y) + I(Z;X)$ and $I(X;Y|Z)$.  
It is not clear to us whether one could hope for such $Z=Z(\Pi_D)$ when $D$ is a \emph{fully adaptive} data structure, 
hence we believe the following question is interesting: 
	\begin{displayquote}
	\centering
	Is it possible to design a r.v. $Z$ where $I(Z;Y), I(Z;X)$ and $I(X;Y|Z)$ are all small?
	\end{displayquote}

\section*{Acknowledgement}{We are grateful to Huacheng Yu for his valuable comments on 
an earlier draft of this paper.}

\bibliographystyle{alpha}
\bibliography{refs.bib}

\appendix

\section{Tight Bounds from \cite{littleadvice}} 
\subsection{Upper Bound} \label{sec:appendixUB}

First, we show that there is a non-trivial $o(n)$-query time semi-adaptive data structure for $\SEL{\DISJ_n}{k}$, derived from a protocol of \cite{littleadvice} given that the querier has access to $i$ and $S_i$ from the beginning.

\begin{theorem}
There exists $t_q = \tilde{O}(\sqrt{n})$ semi-adaptive data structure with $|U| \leq \tilde{O}(\sqrt{n})$.
\end{theorem}
\begin{proof}
Consider $U(\vec{S},T)$ constructed by Chattopadhyay et al. \cite{littleadvice}. They construct $U(\vec{S},T)$ with $|U| \leq \tilde{O}(\sqrt{n})$ independent of $i$ that has the following nice property:
If Alice learns $U$, Alice (just with $i$ and $S_i$) can decode that either $|S_i \cap T| > 1$, therefore $\DISJ_n(S_i,T) = 0$; or learns about potential set of elements in $P \subset (S_i \cap T)$ with $|P| \leq \sqrt{n}$. Alice can transmit this set $P$ using $\sqrt{n} \log n$ bits to Bob, then Bob can announce the answer as $\DISJ_n(S_i,T) = \DISJ_n(P,T)$.


We can translate their protocol to data structure if the querier has access to both $i$ and $S_i$ in the beginning. 

\begin{enumerate}
\item During the update phase, the data structure writes $U$ and $T$ separately. 
\item Querier ($i, S_i$) reads all cells with $U$ using $|U|/w$ queries, then either learns that $\DISJ_n(S_i,T) = 0$ or learns about the candidate set $P \subset S_i$. 
\item Querier checks $\DISJ_n(P,T)$ by accessing corresponding entries for $T$. 
\end{enumerate}

Since $|P| \leq \sqrt{n}$, the total number of queries made is $|P| + (|U| / w ) = \sqrt{n} + (|U| / w ) = \tilde{O}(\sqrt{n})$. Furthermore, the query is semi-adaptive. Querier {\em only} accesses updated cells. The total number of alternation is 1.
\end{proof}

\subsection{Lower Bound} \label{sec:appendixLB}

In this section, we showed that a modified 4-party model subsumes 1.5-round protocol. Recall that a ``1.5-round protocol" \cite{littleadvice} for $\SEL{\DISJ_n}{k}$ proceeds in the following way: 

\begin{enumerate}
\item Charlie sends message $U(\vec{S},T)$ to Bob privately.
\item Bob \emph{forwards} a message $U'(\vec{S},T) \subseteq U$ to Alice (hence this message is \emph{independent} of $i$).
\item Alice sends message $\Pi^{A \rightarrow B}_i$ to Bob, which is dependent on $U', \vec{S}$ and $i$.
\item Bob solves $\DISJ_n(S_i,T)$ from $U, T, i$ and $\Pi^{A \rightarrow B}_i$.
\end{enumerate}

Then \cite{littleadvice} (Theorem 1.2) show a 1.5-protocol for $\SEL{\DISJ_n}{k}$ with $\tilde{O}(\sqrt{n})$ overall communication. They then complement this upper bound with a lower bound when restricted to a 1.5-protocol.
In general, this model is incomparable to our 4-party model. 
However, a 1.5-round protocol can be simulated by a modified 4-party protocol where Charlie is allowed to send Megan a message $U'$ prior to Megan sending messages to Alice and Bob. Formally the 4-party protocol proceeds in following manner

\begin{figure}[h!]
    \centering
    \includegraphics[scale=0.5]{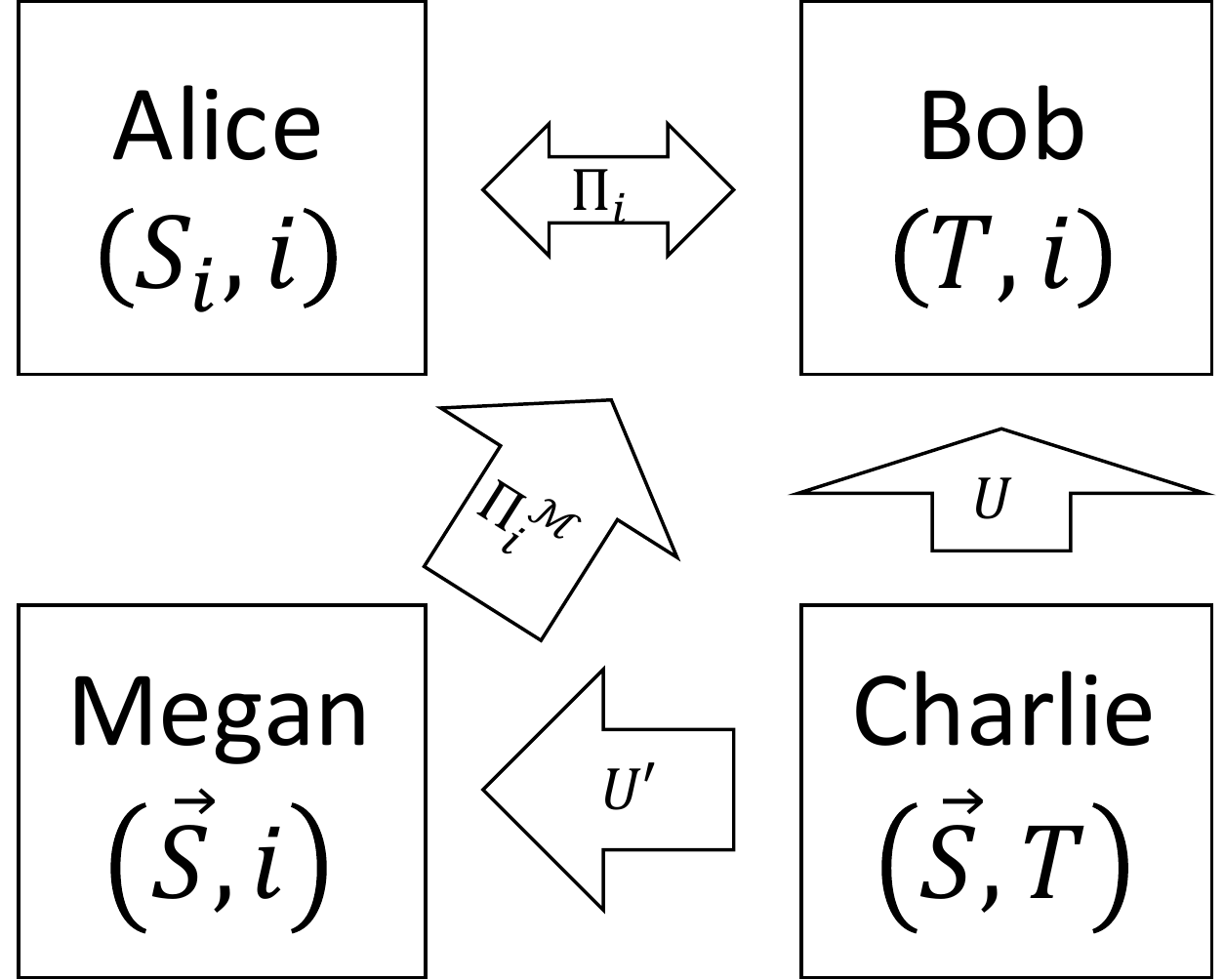}
    \caption{Modified 4-party NOF Communication}
    \label{fig:4party} 
\end{figure}

\begin{enumerate}
\item Charlie sends message $U(\vec{S},T)$ to Bob privately.
\item Charlie sends message $U'(\vec{S},T)$ to Megan.
\item Megan broadcasts $\Pi_i^\cM (\vec{S},i, U')$ to Alice and Bob
\item Alice and Bob communicates and compute $\DISJ_n(S_i,T)$.
\end{enumerate}

Now Megan's message is allowed to depend on $i$, $\vec{S}$ and $U'$. 

\begin{claim}
A 1.5-round protocol where $|U'| < C$ and $|\Pi^{A \rightarrow B}_i| < C$ can be simulated by modified 4-party NOF communication with $|\Pi_i^\cM| < 2C$.
\end{claim}
\begin{proof}
	Under this model, 4-party can then simulate 1.5-protocol via Megan simulating Alice in 1.5-protocol by sending $\Pi_i^{A \rightarrow B}$ and $U'$ as her message $\Pi_i^\cM$. Bob can then decode $\DISJ_n(S_i,T)$ from Megan's message, with no communication between Alice and Bob.
\end{proof}

Now the question is whether the same lower bound for modified 4-party NOF holds as well. Below we argue that our lower bound (\pref{thm:multiphase_informal}) in fact applies to 4-party protocols with Charlie sending a short message to Megan hence also to 1.5-protocols, establishing that our lower bound subsumes \cite{littleadvice} lower bound. 

\begin{theorem}
Any modified 4-party NOF protocol $\Gamma = (U, \Pi_i)$ with $\Pi_i = (U', \Pi_i^\cM, \Pi_i^{A \leftrightarrow B})$ that solves $\SEL{\DISJ_n}{k}$ with $|U| < o(k)$ require $|\Pi_i| > \Omega( \sqrt{n} )$. 	
\end{theorem}
\begin{proofsketch}
This follows from the observation that in such 4-party protocols, \pref{lem:goodq} still holds (up to factor 2) and rest of the proof remains unchanged.
To see why, observe that since $U'$ from step 2 does not depend on the index $i$, Equation \pref{eq:point_of_deviation} in the proof of 
\pref{thm:multiphase_informal} 
can be bounded instead by the following inequality:  

\begin{align*}
& I(T; S_{i_{ < \ell}}, \Pi^{\cM}_{i_{ < \ell}}) \leq I(T; \vec{S}, \Pi^{\cM}_{i_{ < \ell}}, U') \\
& = \underbrace{I(T; \vec{S})}_{=0} + \underbrace{I(T; U'|\vec{S})}_{|U'|} + \underbrace{I(T; 	\Pi^{\cM}_{i_{ < \ell}} | \vec{S}, U')}_{=0} \leq |U'|, 
\end{align*}
where $I(T; \Pi^{\cM}_{i_{ < \ell}} | \vec{S}, U') = 0$ since $U', \vec{S}$ and the index $i$ determine $\Pi^{\cM}_{i}$.

Equation \pref{eq:point_of_deviation} is the only step where the proof used the assumption that Megan speaks first. Now in a 4-party protocol, if we assume further that $|U'| < C$, then by \pref{eq:point_of_deviation}, instead of \pref{eq:smallT} in \pref{lem:goodq}, we get the same bound up to factor 2, i.e.,   $I(Z^{DISJ}; T) < 2C$. Then rest of the proof remains unchanged, hence it shows a $C \geq \Omega( \sqrt{n} )$ lower bound for 4-party protocols. 
\end{proofsketch}

As such, the $\tilde{O}(\sqrt{n})$ 1.5-protocol of \cite{littleadvice} shows that our \pref{thm:multiphase_informal} is in fact tight up to logarithmic factors.

\end{document}